\documentclass[%
 reprint,
 amsmath,amssymb,
 prx,
]{revtex4-1}

\usepackage{amsthm}
\usepackage{graphicx}% Include figure files
\usepackage{dcolumn}% Align table columns on decimal point
\usepackage{amssymb}
\usepackage{bm}% bold math
\usepackage[colorlinks, citecolor=blue]{hyperref}
\hypersetup{colorlinks,linkcolor=blue,citecolor=blue,filecolor=blue,urlcolor=blue}
\usepackage{physics} % Physics macros
\usepackage{qcircuit} % Quantum circuits
\usepackage{tikz} % TikZ
\usetikzlibrary{shapes.geometric}
\newtheorem{theorem}{Theorem}

\newtheorem{definition}{Definition}

\newtheorem{lemma}{Lemma}
\newtheorem{lem}{Lemma}
\let\emptyset\varnothing

\newcommand{\innerp}[2]{\langle#1|#2\rangle}

\begin{document}

%\preprint{APS/123-QED}

\title{Enhancing Generative Models via Quantum Correlations}% Force line breaks with \\

\author{Xun Gao$^{1}$, Eric R. Anschuetz$^{2}$, Sheng-Tao Wang${}^{1,3}$, J. Ignacio Cirac${}^{4,5}$ and Mikhail D. Lukin${}^{1}$}  
\affiliation{${}^1$Department of Physics, Harvard University, Cambridge, Massachusetts 02138, USA}
 \affiliation{${}^2$MIT Center for Theoretical Physics,
77 Massachusetts Avenue, Cambridge, MA 02139, USA}%Lines break automatically or can be forced with \\
 \affiliation{${}^3$QuEra Computing Inc., Boston, Massachusetts 02135, USA}
\affiliation{${}^4$Max-Planck-Institut f{\"u}r Quantenoptik, Hans-Kopfermann-Str.\ 1, 85748 Garching, Germany}
\affiliation{${}^5$
Munich Center for Quantum Science and Technology (MCQST), Schellingstr. 4, D-80799 M{\"u}nchen, Germany}


\begin{abstract}

Generative modeling using samples drawn from the probability distribution constitutes a powerful approach for unsupervised machine learning. Quantum mechanical systems can produce probability distributions that 
exhibit quantum correlations which are difficult to capture using classical models. 
We show theoretically  that such quantum correlations provide a powerful resource for generative modeling. In particular, we provide an unconditional proof of separation in expressive power between a class of widely-used generative models, known as Bayesian networks, and its minimal quantum extension. 
We show that this expressivity advantage is associated with quantum nonlocality and quantum contextuality. Furthermore, we numerically test this separation on standard machine learning data sets and show that it holds for practical problems. The possibility of quantum advantage 
demonstrated in this work not only sheds light on the design of useful quantum machine learning protocols but also provides inspiration to draw on ideas from quantum foundations to improve
purely classical algorithms.  
\end{abstract}

\pacs{Valid PACS appear here}% PACS, the Physics and Astronomy
\maketitle

%\tableofcontents

\section{Introduction}

%\subsection{Motivation}

Over the past three decades, the field of machine learning has achieved remarkable success. A variety of powerful models and algorithms have been developed and deployed for broad applications ranging from computer vision and natural language processing to 
autonomous vehicles ~\cite{shalev2014understanding,bishop2006pattern,goodfellow2016deep}. Unsupervised learning, involving the task of learning from unlabeled data sets, is among the frontier areas of machine learning research. This task is typically much more challenging than supervised learning. The most common approach to tackle unsupervised learning problems is \emph{generative modeling},
where one attempts to construct and train models with efficient representations for high-dimensional probability distributions.  
One of the most important aspects of any generative model is its expressive power, which, together with associated training algorithms, primarily determines the model performance. Models with high expressive power can capture complex correlations in the target probability distribution, while upholding the standard wisdom of Occam's razor by keeping the structure simple (typically corresponding to  a simple connectivity structure or limited number of  parameters).

%%%%

Quantum systems are known to produce complex probability distributions that are hard to capture with classical generative models~\cite{google_supremacy}. For this reason,  quantum models are believed to be more powerful in tackling unsupervised learning tasks. %Motivated by these considerations,  
Consequently, over the past few years, quantum machine learning has emerged as a promising approach to enhance machine learning performance. However, apart from abstract computational complexity arguments~\cite{gao2018quantum,coyle2020born}, any potential quantum advantage in quantum machine learning models and its physical origin is not well understood.   
%physical understandings remain scarce.
%although supported by intuition from quantum interference and evidence from computational complexity theory, still remains unclear~\cite{gao2018quantum}. 
Motivated by these considerations, in this work we explore the role of 
quantum correlations associated with  non-locality and contextuality~\cite{einstein1935can,bell1964einstein,bell1966problem,kochen1975problem},
that are known to 
%are the most fundamental features unique to quantum mechanics and they 
%have been shown to 
be the key resource for quantum advantage in many quantum information processing tasks 
%such as communication and  quantum computing
~\cite{anders2009computational, buhrman2010nonlocality,howard2014contextuality,bermejo2017contextuality,frembs2018contextuality}, in unsupervised machine learning problems such as natural language processing. Intuitively, the potential for quantum advantage for such tasks can be understood by noting that, in language processing problems, one often needs to read the whole sentence to understand the meaning of some words; in other words, the interpretation of a word might depend on the context, which shares similarity with observables in quantum contextuality. 
%Thus we may expect complex quantum correlations such as contextuality to play a role. % Thus motivated, 
In this work, we explore if and how generative models could benefit from such quantum correlations.

Specifically, we focus on a class of standard generative models, known as Bayesian networks, and show that  quantum correlations can be used to achieve provable separation between such models and their minimal quantum extension described by  a corresponding class of tensor networks. Focusing on sequential models, we compare subclasses of Bayesian networks with the corresponding 1D tensor networks described by  Matrix Product States (MPSs) and show that MPS features more expressive power compared to traditional machine learning models \cite{perez2006matrix,schollwock2011density,stoudenmire2016supervised,PhysRevX.8.031012,1803.09111,glasser2020probabilistic}. Since the 1D models can be efficiently evaluated on a classical computer, we also numerically test the models on real-world data sets and find an improvement in generative modeling using MPS. While these results provide new insights into the power of MPS-based machine learning algorithms, since there exists a subclass of tensor networks that cannot be efficiently simulated using a classical computer but can be implemented on a quantum computer, our results also suggest the possibility of a quantum advantage in generative machine learning.

Our paper is organized as follows. In the next section, we provide an outline of the main results and discuss their implications.  In Sec.~\ref{sec:bqc}, we  review Bayesian networks and their quantum circuit interpretation, and introduce our minimal quantum extension of Bayesian networks and its relation with tensor networks. In Sec.~\ref{sec:correlations}, we prove separations in expressivity between the two classes of models in learning sequential data sets. In Sec.~\ref{sec:numerics}, we give numerical evidence that this separation often holds not only in theory but also in practice, by showing separations on a variety of standard machine learning data sets. Finally, in Sec.~\ref{sec:outlook} we discuss the implications of our results and consider future lines of research.

% Since quantum correlations are characterized by correlations in the probability distribution resulting from measuring the quantum system, we ask the natural question: could generative models benefit from these quantum correlations?

% should it be a separate subsection, or merged with the previous one?

\begin{figure*}[tp]
\includegraphics[width=\textwidth]{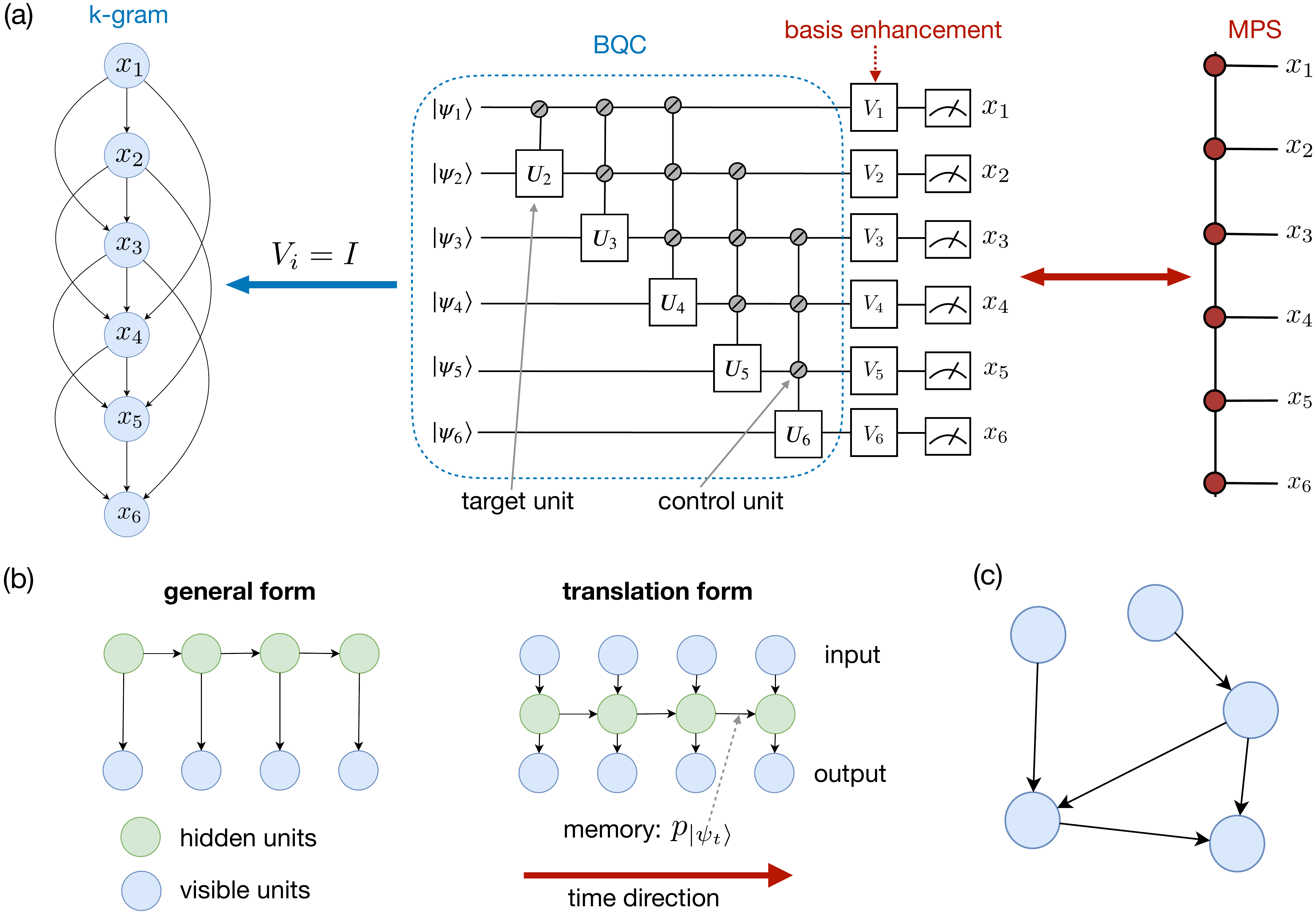}
\caption{A Bayesian network and its quantum circuit. (a) $k$-gram model and its basis enhancement. The leftmost diagram is a 4-gram model as the transition probability takes the form $p\left(x_l\mid x_{l-1},x_{l-2},x_{l-3}\right)$, involving 4 variables. All of the variables are visible. The middle diagram is a basis-enhanced Bayesian quantum circuit (BBQC). The part in the blue dashed box is a Bayesian quantum circuit (BQC) of the corresponding Bayesian network since measuring the output qubits in the computational basis results in the same probability distribution. The basic elements of BQCs are uniformly controlled gates of which control units and target units are labeled. The BBQC is a special case of tensor networks. In this case, the BBQC can be written as a Matrix Product State (MPS) as shown in the rightmost diagram. (b) A Hidden Markov model (left) and a hidden Markov model for a translation problem (right). Visible variables and hidden variables are colored as blue and green respectively. The top and bottom visible units store input and output respectively. We adopt a dynamical point of view: the HMM can be understood as a measurement-driven evolution of probability distributions which encode the quantum states at the $t$-th time step over the hidden variables in the $t$-th virtual bond. (c) An example of a Bayesian network on a general directed graph. More detailed discussion is in Appendix~\ref{app:general_graph}.}
\label{fig:baynet}
\end{figure*}

\section{Summary of Results and Their Implications}\label{sec:summary}

%To be specific, we consider a minimal extension of 
Bayesian networks, associated with  a class of generative models based on directed graphs, have a wide range of applications \cite{niedermayer2008introduction}. Probability distributions described by Bayesian networks are known to have an equivalent formulation in the computational basis measurements of a class of quantum circuits known as \emph{Bayesian Quantum Circuits} (BQCs, see Fig.~\ref{fig:baynet} and Ref.~\cite{low2014quantum}). By extending this class to allow local measurement beyond the computational basis, we define a class of quantum-inspired models dubbed \emph{Basis-enhanced Bayesian Quantum Circuits} (BBQCs), which are a special class of tensor networks that inherit the graph structure of their corresponding Bayesian networks.

In this work, we construct BBQCs that have unconditional expressivity separations compared to their classical counterparts, i.e.\ Bayesian networks on the same directed graphs. Instead of requiring an exact representation, we relax the comparison criterion to allow for any finite error in the forward and backward Kullback--Leibler (KL) divergence. This is equivalent to the condition \begin{equation}\label{eq:error}
q\left(\bm{x}\right)=0\iff p\left(\bm{x}\right)=0, \forall\bm{x}.
\end{equation}
where $p$ and $q$ are the two comparison distributions. This error model, however, is still not practical enough; for instance, when $q\left(\bm{x}\right)=0$ and $p\left(\bm{x}\right)$ is very small, the exact KL divergence is infinite. In this paper, we adopt this error model to obtain rigorous proofs, but show numerically that there exists a finite advantage in KL divergence even when in practical training $q\left(\bm{x}\right)$ does not have exact zero probabilities for any $\bm{x}$. KL divergence is a widely used error model in unsupervised machine learning.

We first analyse the implications of quantum nonlocality for a so-called $k$-gram model, a very successful Bayesian network model used in natural language processing (see Fig.~\ref{fig:baynet}(a)) . We introduce a basis-enhanced 2-gram model, shown in Fig.~\ref{fig:quantum2gram} (where the left is the BBQC and the right is the corresponding Bayesian network), and prove that any $k$-gram model for $k<\left\lfloor(n-1)/2\right\rfloor$ cannot approximate its probability distribution under finite KL divergence. %The approximate KL divergence is widely used in unsupervised machine learning; though we require the exact KL divergence to obtain rigorous proofs, we show numerically that there appears to be a finite advantage even in the approximate case. 
The proof makes use of  
%a well-known property of quantum nonlocality: 
the nonlocal correlations present in measuring a GHZ state that cannot be described by local hidden variable models~\cite{greenberger1990bell}. We extend this argument to a cluster state where the qubits are measured either in the $X$ or $Y$ basis. This state can be represented by a basis-enhanced 2-gram model, but not a local hidden variable model. By measuring all the qubits other than the first, middle, and final qubits (shown as dashed circles in Fig.~\ref{fig:quantum2gram}(a)), the state reduces to a GHZ state up to local unitaries. For the corresponding $k$-gram model, however, the conditional probability distribution factorizes and can be described by a local hidden variable model. This result is summarized as the following theorem:
\begin{theorem}[$k$-gram models and quantum non-locality]\label{thm:kgram}
There exists a family of basis-enhanced $2$-gram models with generated probability distribution $q$ such that any classical $k$-gram models with $k=o\left(n\right)$ (where $n$ is the length of the 2-gram model) cannot approximate $q$ to the error model in Eq.~\eqref{eq:error}. This separation originates from quantum nonlocality.
\end{theorem}
%\noindent This unconditional separation in expressive power also holds under the $l_1$-norm error model.

Since $k$-gram models can only capture local correlations, we then investigate a more expressive class of models, \emph{hidden Markov models} (HMMs, shown in Fig.~\ref{fig:baynet}(b)), which are widely used in reinforcement learning and temporal pattern recognition. HMMs extend $k$-gram models by introducing hidden variables as memory to capture long-range correlations, and they are the most generic 1D sequential generative models, including both feedforward and recurrent neural networks (given finite precision) as specific instances. We focus on the HMMs in the so-called translation form, with input and output regarded as original and target languages, respectively, as shown in Fig.~\ref{fig:baynet}(b). Basis-enhanced versions of such HMMs correspond to  a special instance of Matrix Product Operators (MPO). We use quantum contextuality to prove an expressivity separation between classical HMMs and their basis-enhanced counterparts. Specifically, we prove the following theorem:
\begin{theorem}[Hidden Markov models and quantum contextuality]\label{thm:hmm}
There exists a family of basis-enhanced $2$-gram models with a state space of dimensionality $D$, that cannot be approximated, in the sense of Eq.~\eqref{eq:error}, by any classical hidden Markov models in the translation form (Fig.~\ref{fig:baynet}(b)) with a number of hidden units fewer than $D^{\Omega(\log D)}$. The separation originates from quantum contextuality.
\end{theorem}
\noindent Here, the quantum-enhanced model is based on a $2$-gram model (Fig.~\ref{fig:Clifford}(c)), which is a special case of an HMM. The corresponding quantum circuit representation is shown in Fig.~\ref{fig:Clifford}(a-b), with  $D=2^n$ where $n$ is the number of qubits.

%show that there exists a basis-enhanced HMM with $D$ hidden states, corresponding to a bond dimension of $D$ in the MPO representation, such that any HMM translation form of model that approximates this model must have at least $D^{O(\log D)}$ hidden states.

This result can be understood by considering  
%through a dynamics point of view, where 
the 1D structure of the models  as a time dimension as shown in Fig.~\ref{fig:baynet}(b). The state of the HMM (or the corresponding basis-enhanced $2$-gram model) is encoded as a probability distribution $p_{\ket{\psi_t}}$ (quantum state $\ket{\psi_t}$) over the hidden states of the HMM (virtual bond of the MPO) at the $t$-th time step. The number of hidden states (bond dimension) corresponds to the memory of the system, of which the logarithm is the number of bits (qubits) of memory required to store the state of the system. The inputs and outputs are different measurement basis and measurement results, respectively. In order to simulate the quantum process, the HMM should have enough memory of the previous measurement basis and measurement results to predict future behavior. Within this picture,  the translation form of HMMs is essentially equivalent to hidden variable models (also called ontological models)~\cite{harrigan2007representing,karanjai2018contextuality}. Quantum contextuality formalizes the phenomenon that a measurement result % (determined by a hidden variable) 
of an observable should depend on which commuting observable set (known as a context) the observable belongs to in the given measurement scenario. However, since there are many different commuting sets that include this observable, when it is measured, a hidden variable must memorize which context this observable belongs to in any given measurement scenario. A well-known example of contextuality is associated with 
%As an example, consider the three states that compose 
the Mermin--Peres magic square~\cite{mermin1993hidden,peres1991two}. Our proof strategy for Theorem~\ref{thm:hmm} relies on showing that Mermin--Peres magic squares are very common in stabilizer states~\cite{gottesman1997stabilizer}, and we use that to find a lower bound on the number of hidden states needed to accurately represent stabilizer measurements.

Finally, we 
%step away from the specific models used in our analysis and 
evaluate the relative performance of BBQCs and Bayesian networks on standard machine learning data sets. We focus on the relative performance of HMMs and their basis-enhanced counterparts, but here we use the general HMM graph structure in Fig.~\ref{fig:baynet}(b). As basis-enhanced HMMs are a special case of MPSs, we are able to evaluate the expressive performance of both HMMs and basis-enhanced HMMs efficiently on a classical computer. Specifically, we evaluate both models on the biofam data set~\cite{Mueller2007,ritschard2013exploratory}, which is known to be well modeled by a simple $2$-gram model. Additionally, we evaluate both models on the more difficult SPECT Heart and Promoter Gene Sequences data sets~\cite{Dua:2019}.
% To accomplish this, we train the HMM using the well known Baum--Welch algorithm~\cite{baum1970}, and the basis-enhanced HMM using gradient descent, such that the distribution of the models approximately match those of the data sets. 
%As we are interested in the absolute performance capabilities of the models, we pick the best-performing trained model over ten training instances, and compare their raw performance in terms of the achieved KL divergence on both the training data and the withheld testing data; observing good performance on the latter signals that the trained model is not overfitting to the training data. 
We find that the basis-enhanced HMM outperforms the HMM on both the training and testing data for both the SPECT Heart and Promoter Gene Sequences data sets, and achieves comparable performance on the control biofam data set. These results are summarized in Fig.~\ref{fig:performance_plot}. In addition, we perform a likelihood-ratio test on the goodness of fits of the two models; this measures the statistical significance of the expressivity gap of the two models, accounting for the potential overfitting of the basis-enhanced model due to its having more parameters. We show that the improvement in KL divergence of the basis-enhanced HMM over the HMM is statistically significant on the SPECT Heart and Promoter Gene Sequences data sets to a confidence of $>5\sigma$. These results are summarized in Fig.~\ref{fig:p_value_plot}.

% In particular, we focus on a class of quantum machine learning models based on Bayesian networks, which have a natural quantum circuit interpretation, where we extend the model by allowing general local measurements. Furthermore, for certain subclasses of Bayesian networks our quantum extensions are classically simulable, and therefore our results also prove separations in expressivity between certain classes of classical machine learning models. We also give numerical evidence that this separation in expressivity often holds in practice, by training on standard machine learning data sets and showing a separation in the achieved log-likelihood on both the training and on withheld testing data.

%\subsection{Implications of Results}

Our results have important implications for developing both classical and quantum machine learning methods. 
%As discussed above, we directly showed both analytically and numerically that minimal extensions to classical Bayesian networks can lead to large improvements in the expressive performance of the model. 
Although the source of the advantage mechanisms described above is inspired by quantum correlations, for many classes of Bayesian networks---including $k$-gram and hidden Markov models---our extension still results in classical models, described by  
special cases of MPS or MPO that can be efficiently implemented on classical systems. 
In such cases, our results indicate that with a minimal computational overhead, one can obtain markedly improved modeling of data using novel quantum-inspired classical approaches. 
While a number of classical machine learning techniques are already employing the methods based on tensor network \cite{,PhysRevX.8.031012,glasser2020probabilistic,NIPS2019_8429,stoudenmire2016supervised}, our results demonstrate that one could draw on ideas from quantum foundations to show unconditional efficiency separations for such novel classical models.
%We hope this could offer inspiration for future studies.

Furthermore, for more complicated models, such as 2D Bayesian networks, where our extension cannot be efficiently implemented on a classical computer, our results provide important insights into designing novel quantum machine learning algorithms. We emphasize that in contrast 
with the previously proposed quantum machine learning models, 
which consider generic quantum circuits to provide  quantum correlations, our approach makes use of the minimal extension of classical models.  
%known to give quantum advantages in other contexts. 
This is important, since unstructured quantum circuits are not practical machine learning models due to difficulties associated with  barren plateaus in training landscape~\cite{McClean_2018,2design,marrero2020entanglement} and the no free lunch theorem~\cite{poland2020no}. By restricting our study to minimal quantum extensions of classical machine learning models, we sidestep these issues while still maintaining a quantum advantage over the corresponding classical model. In addition, this minimal approach allows us to understand the origin of quantum advantage, which is essential for efficient 
design of new quantum machine learning models. In particular, our results indicate
%In short, we show an unconditional quantum advantage in generative modeling, and 
%give evidence 
that a practical computational advantage may be obtained on quantum devices in machine learning tasks perhaps in the near future.

\section{Unsupervised Generative Modeling and Minimal Quantum Extensions}\label{sec:bqc}

%One way to understand 
Many unsupervised machine learning tasks can be understood through a probabilistic lens. In this approach the data $\bm{x}$ (e.g.\ $\bm{x}$ could be a vector $\left(x_1,\cdots, x_n\right)$ representing pixels of a handwritten digit) are regarded as being generated identically and independently from an unknown probability distribution $p_{\mathcal D}$ 
%(i.i.d.\ assumption)
~\cite{shalev2014understanding,goodfellow2016deep}. The task of unsupervised learning is to characterize some aspects of this distribution explicitly or implicitly. \emph{Generative models} attempt to represent the entire probability distribution $p_{\mathcal D}$ approximately, thus providing an almost complete characterization of $p_{\mathcal D}$.

Directly representing a probability distribution over $\left(x_1,\cdots, x_n\right)$ requires a number of parameters exponential in $n$. However, assuming some underlying structure on the distribution $p_{\mathcal{D}}$, it is expected that only a polynomial number of parameters is sufficient to approximate  $p_{\mathcal D}$ for most natural distributions. One can draw an analogy to quantum many-body physics: physical states, which play the role of naturally occurring distribution $p_{\mathcal D}$, typically require a polynomial number of parameters to represent, whereas generic states in the entire Hilbert space, like generic distributions with $n$ variables, need an exponential number of parameters~\cite{poulin2011quantum}. Graphical structures with a polynomial number of parameters often constitute efficient representations for generative models, similar to the representation of tensor networks in quantum many-body physics~\cite{orus2014practical,bridgeman2016hand,verstraete2008matrix,vidal2007entanglement,verstraete2004renormalization}.

In what follows,  we focus on a particular type of probabilistic graphical models, called Bayesian networks, and explore one minimal quantum extension of this classical model. This allows us to understand the origin for the underlying quantum advantage, which sheds light on the design of quantum models. 
%The extended quantum model can also be understood as a subclass of tensor network models, which may not be efficiently calculated on a classical computer. 
We  emphasize that our approach is not limited to Bayesian networks and can be extended to other models. 

% Given a model distribution $q$, perhaps the most common cost function used to measure the effectiveness of $q$ at modeling $p$ is the Kullback--Leibler (KL) divergence~\cite{kullback1951information,goodfellow2016deep}:
% \begin{equation}\label{eq:KL}
% D_\text{KL}\left(p\mid\mid q\right)=\sum_{\bm x}p\left(\bm x\right)\log\frac{p\left(\bm x\right)}{q\left(\bm x\right)},
% \end{equation}
% which is non-negative and saturates to 0 when $p=q$. In this paper, we use it to measure the divergence between two generative models in Sec.~\ref{sec:correlations} and also use it to measure the the divergence of data distributions from generative models as a cost function for training in Sec.~\ref{sec:numerics}. 

\subsection{Bayesian Networks and Language Processing}\label{sec:bn}

Bayesian networks are a class of generative models that define a probability distribution through a directed acyclic graph in the following way (see an example in Fig.~\ref{fig:baynet}(a,c)): for each node $x_i$ (associated with a random variable $x_i$), we assign a transition probability $p\left(x_i\mid\text{parents of }x_i\right)$, where the parents of $x_i$ are nodes with edges directed towards $x_i$; if there is no parent node for $x_i$, the transition probability reduces to the marginal probability $p(x_i)$; then, the product of these transition (marginal) probabilities 
\begin{equation}
p\left(x_1,\cdots,x_n\right)=\prod_i p\left(x_i\mid\text{parents of }x_i\right)
\end{equation}
is the final joint probability distribution. 
% Here we abuse the notation of the variable associated with a node, $x_i$, as the notation of the corresponding node.

% k-gram model, hidden Markov model, translation problem
Bayesian networks are useful in natural language processing as statistical language models. Roughly speaking, statistical language models are generative models for language and they are used to generate a probability distribution of ``meaningful'' combinations of word sequences. A good design for statistical language model is crucial to the performance of machine learning for natural language processing such as translation, speech recognition, and natural language generation~\cite{martin2009speech}.

%\begin{figure}
%\centering
%\includegraphics[width=1\linewidth]{figures/kgram.pdf}\caption{A $k$-gram model. In this case, $k=4$ because the transition probability has the form of $p(x_l|x_{l-1},x_{l-2},x_{l-3})$, which involves 4 variables. All the variables are visible.
%}%
%\label{fig:kgram}%
%\end{figure}

Historically, prior to the rise of deep learning, one of the most commonly used statistical language models were \emph{$k$-gram models}~\cite{martin2009speech}, which are Bayesian networks on a 1D graphical structure with $k-1$ neighbors connected (see Fig.~\ref{fig:baynet}(a) for an example with $k=4$). Despite their simplicity, certain types of generative neural networks can also be viewed as $k$-gram models, e.g., deep belief nets~\cite{hinton2009deep} even for $k=2$ (see Appendix~\ref{app:relation_models}).

In order to capture more complex correlations, a more complex model, the \emph{hidden Markov model} (HMM), with additional hidden nodes on a 1D graphical structure with nearest-neighbor connections, is introduced as shown in the translation form of Fig.~\ref{fig:baynet}(b). $k$-gram models are a special case of HMMs (see Appendix~\ref{app:relation_models}), as HMMs can store $k$-length correlations into the hidden variables. The graph structure of the HMM is easily generalized for translation problems. From a probabilistic point of view, a translation problem can be considered as a modeling problem for a conditional probability distribution $p\left(\bm{y}\mid\bm{x}\right)$ by a generative model, e.g.\ using the HMM shown in the translation form of Fig.~\ref{fig:baynet}(b) where $\bm{x}$ (represented by the top row visible variables) is a ``sentence'' in the original language and $\bm{y}$ (represented by the bottom row visible variables) corresponds to the translation in the target language with the conditional probability $p\left(\bm{y}\mid\bm{x}\right)$. If the prior probability $p\left(\bm{x}\right)$ can be captured by an HMM in the general form of Fig.~\ref{fig:baynet}(b), the joint probability distribution $p\left(\bm{x},\bm{y}\right)$ could be viewed as a special case of the general form of Fig.~\ref{fig:baynet}(b).

%\begin{figure}
%\centering
%\includegraphics[width=0.9\linewidth]{figures/HMM.pdf}\caption{(a) Hidden Markov model. (b) Hidden Markov model for translation problem. Visible variables and hidden variables are colored as blue and green respectively.}%
%\label{fig:HMM}%
%\end{figure}

\subsection{Bayesian Quantum Circuits and Minimal Quantum Extensions}
% most of the result should be put on appendix
The key to defining our quantum extension of Bayesian networks is the equivalence between Bayesian networks and a restricted class of quantum circuits, which we call \emph{Bayesian quantum circuits} (BQCs; see also Ref.~\cite{low2014quantum}). %The exact mapping can be found in Appendix~\ref{app:equivalence}.
BQCs are defined such that the probability distribution sampled from the quantum circuits, by measuring the visible qubits in the computational basis, is the same as the probability distribution defined by the corresponding Bayesian network. In addition, we define a minimal quantum extension of Bayesian networks, \emph{basis-enhanced Bayesian quantum circuits} (BBQCs), by allowing the final measurements to be in an arbitrary local basis.

\subsubsection{Bayesian Quantum Circuits}

The building block of BQCs are uniformly controlled gates. A uniformly controlled gate is a generalization of a control-$U$ gate, which consists of $k$ control qubits and 1 target qubit~\cite{bergholm2005quantum}: if the control qubits are in the state $\ket{x_1\cdots x_k}$, the target qubit will be applied by a unitary $U\left(x_1\cdots x_k\right)$ (i.e.\ a single qubit unitary determined by $x_1\cdots x_k$). For convenience, we introduce the concepts of control units and target unit for a uniformly controlled gate as shown in Fig.~\ref{fig:baynet}(a). 
\begin{definition}[Bayesian quantum circuits]\label{def:BQN}
A Bayesian quantum circuit consists of a sequence of uniformly controlled gates followed by a measurement of a subset of the qubits in the computational basis, with the following restrictions:
\begin{itemize}
    \item Each uniformly controlled gate only targets a single qubit, reflecting the fact that there is only a single target variable in a transition probability in Bayesian networks.
    \item After being used as a control qubit, the qubit cannot be targeted by a uniformly controlled gate, reflecting the fact that Bayesian networks are defined on directed acyclic graphs.
\end{itemize}
\end{definition}

The exact mapping between Bayesian networks and BQCs can be found in Appendix~\ref{app:equivalence}. The implementation of an arbitrary uniformly controlled gate by elementary gates is in general not efficient, %since it requires the computation of arbitrary boolean functions, 
since it typically consists of an exponential number of standard control gates~\cite{bergholm2005quantum}. However, for most relevant Bayesian networks, the transition probabilities only involve a few variables or are highly structured, which will make the corresponding uniformly controlled gates easy to implement. We discuss further the implementation of BQCs with multi-qubit collective gates in Appendix~\ref{app:MRG}.

\subsubsection{Basis-Enhanced Bayesian Quantum Circuits}
\label{sec:bmbn_features}

As defined above, BQCs can only produce distributions that correspond to Bayesian networks and thus there is no quantum advantage in the expressivity of the model. In principle, there are many possible ways to generalize this model, such as violating the order requirement between target and control units or generalizing the uniformly controlled gates to more general gates. These generalizations will include universal quantum circuits and thus lose resemblance with classical Bayesian networks. To identify the differences between quantum models and their classical counterparts in terms of quantum advantage, we introduce a natural, minimal extension by allowing the measurements to be in other local basis beyond the computational basis. We call this \emph{basis-enhanced Bayesian quantum circuits} (BBQCs). Note that the locality in measurement basis is important; otherwise, the model will be as powerful as universal quantum circuits. 
\begin{definition}[Basis-enhanced Bayesian quantum circuits]
A basis-enhanced Bayesian quantum circuit is a generalization of Bayesian quantum circuit, where the measurements can be in any local basis beyond the computational basis.
\end{definition}

This seemingly modest extension of classical Bayesian networks comes with considerable quantum advantages. For general underlying Bayesian networks, it can be shown that the quantum extension has an exponential improvement in expressive power compared to any ``reasonable'' classical generative models based on computational complexity assumptions (see Appendix~\ref{app:general_graph} for the rigorous proof). However, certain aspects of the complexity proof are unsatisfying. First, it relies on unproven computational complexity assumptions. Second, it does not provide physical insights and understanding on what gives rise to the purported quantum advantage. Therefore, in Sec.~\ref{sec:correlations}, we show explicitly that quantum correlations are the source of quantum advantage for BBQCs. The unconditional separation between classical and quantum models based on quantum correlations is, however, more modest than the separation guaranteed by the complexity-theory based arguments. The analysis can be potentially generalized to other models beyond Bayesian networks as discussed in Sec.~\ref{sec:beyondBayesian}.

\subsubsection{BBQCs as Tensor Networks}

Here, we remark that the quantum model BBQC is still a special case of tensor networks. We can use the $k$-gram model as an example as shown in Fig.~\ref{fig:baynet}(a). Clearly, the BQC is a tensor network. Since the qubits are arranged on a line, we regard this tensor network as an MPS. The bond dimension is bounded by $2^{k-1}$. The exponent is the maximum amount of information transmitted through a qubit and the unitary to change the measurement basis does not increase the bond dimension. Generally speaking, if the degree of the graph is bounded, the bond dimension around a qubit in the tensor network is also bounded. Therefore, a BBQC can still be understood as a tensor network.

\section{Provable Expressivity Separation Through Quantum Correlations}\label{sec:correlations}

To demonstrate how quantum correlations give rise to quantum advantages, we compare the power of BQCs and BBQCs in generating sequential data. We show that, at least for some toy models, several fundamental non-classical characteristics of quantum theory, i.e. non-locality and contextuality, can be used as resources of quantum advantage for generative models. At the same time, BBQCs are special classes of tensor networks. Our proof also demonstrates that ideas from quantum correlations can be used to show unconditional separation between purely classical models ($k$-gram versus MPS or HMM versus MPO). 
\subsection{Error Models}\label{sec:errors}
Given a target probability distribution $p_{\mathcal D}$ and a distribution generated by a generative model $p$, one of the most commonly used cost function to measure the effectiveness of $p$ at modeling $p_{\mathcal D}$ is the forward KL divergence~\cite{kullback1951information,goodfellow2016deep}:
\begin{equation}\label{eq:KL}
D_\text{KL}\left(p_{\mathcal D} \mid\mid p\right)=\sum_{\bm x}p_{\mathcal D}\left(\bm x\right)\log\frac{p_{\mathcal D}\left(\bm x\right)}{p\left(\bm x\right)},
\end{equation}
which is non-negative and lower-bounded by 0 when $p=p_{\mathcal D}$. Since KL divergence is asymmetric, one may also consider the reverse KL divergence, $D_\text{KL}\left(p \mid\mid p_{\mathcal D}\right)$. The choice of forward versus reverse KL divergence when training unsupervised learning models reflects different priorities in the trained distribution~\cite{goodfellow2016deep}. 

%In this paper, we use it to measure the divergence between two generative models in Sec.~\ref{sec:correlations} and also use it to measure the the divergence of data distributions from generative models as a cost function for training in Sec.~\ref{sec:numerics}. 

To compare the expressive power of classical versus quantum models in generative modeling, we use the KL divergence to measure how effective classical models can generate a distribution originating from the corresponding minimally-extended quantum model. In particular, we denote the probability distribution generated from BQCs as $p$ and from BBQCs as $q$, and we investigate what is the separation between $p$ and $q$ in terms of expressive power. 

The error model we use in the following theoretical analysis is to require both the forward KL divergence $D_\text{KL}\left(q \mid\mid p\right)$ and reverse KL divergence $D_\text{KL}\left(p \mid\mid q\right)$ to be finite. This is equivalent to Eq.~(\ref{eq:error}).
Thus, $p$ approximating $q$ under this error model is a weaker requirement than a small divergence of $p$ from $q$. Nevertheless, we now show %, in the rest of Sec.~\ref{sec:correlations},
that quantum correlations, such as nonlocality and contextuality, give rise to quantum advantages. 

Before proceeding, we note that various other common error models can be considered. One widely used one, multiplicative error, is a stronger requirement and a less realistic error model than the finiteness of KL divergence. This is used in our complexity theory based proof of quantum advantage in Appendix~\ref{app:general_graph}. In Sec.~\ref{sec:beyond}, we discuss additional error models that are more realistic and more robust to small perturbations of model parameters. Relations among error models are explained in Appendix~\ref{app:errors}.

% $D_{\text{KL}}\left(p\mid\mid q\right)$ being finite implies closeness under this model. Thus $p$, approximating $q$ under this error model is a weaker requirement than a small divergence of $p$ from $q$. For the rest of Sec.~\ref{sec:correlations}, we will adopt this error model to separate the expressivities of Bayesian networks and basis-enhanced Bayesian networks. Robustness under this error model also implies robustness under the KL divergence.

\subsection{A Toy Model: $k$-gram Models and Quantum Nonlocality}

\begin{figure}
\centering
\includegraphics[width=1\linewidth]{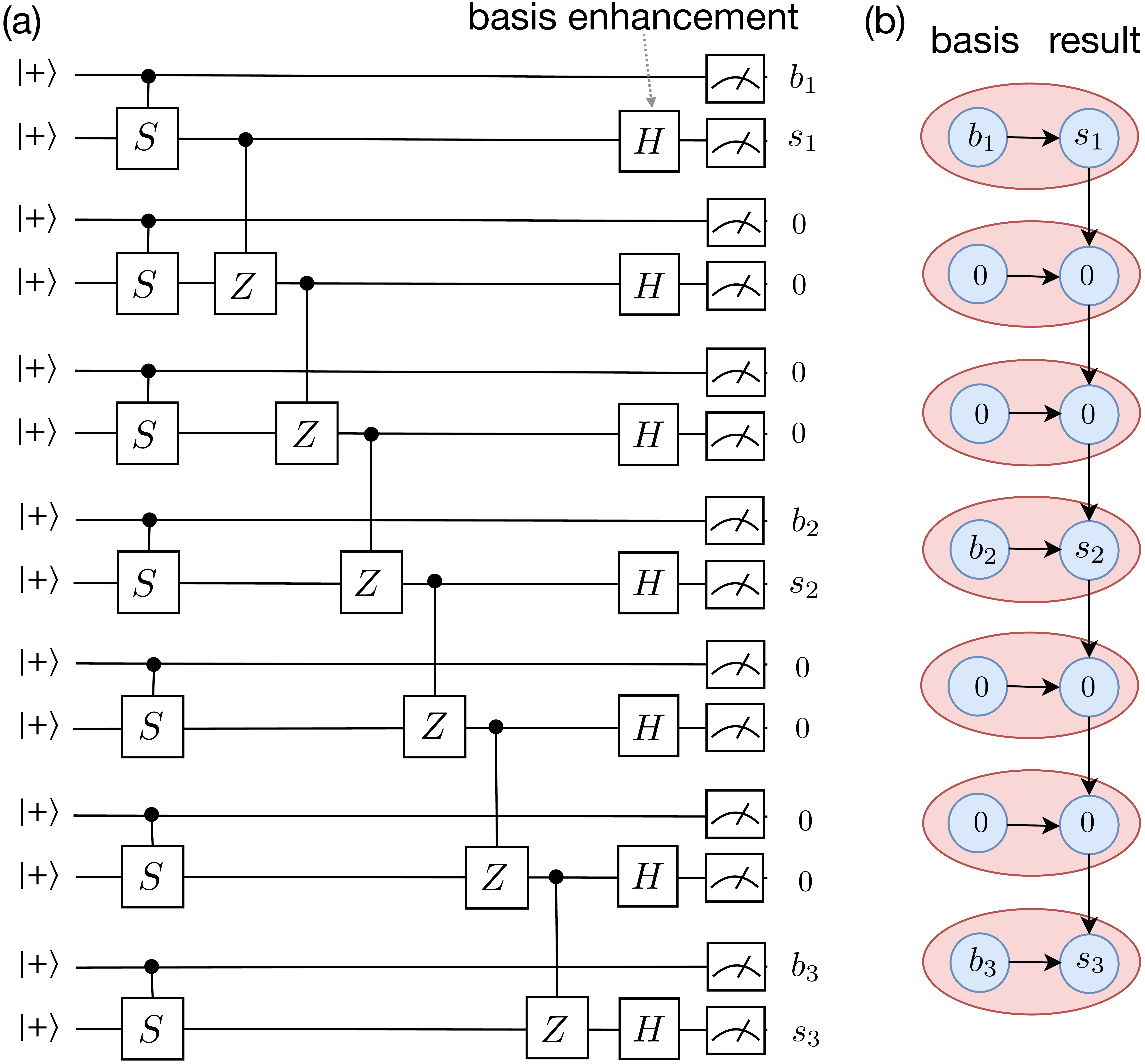}\caption{
(a) An example of a quantum $2$-gram model. In this example, each pair of 2 qubits (counting from the top) corresponds to 1 variable. This BBQC can produce a ($4+3$)-qubit cluster state and measure each qubit on the cluster state in either $X$ or $Y$ basis. For each pair, the first qubit is used to choose the measurement basis, i.e., if it is measured to be 0, the second qubit in this pair will be measured in the $X$ basis otherwise in $S^\dag XS =Y$ basis. The second qubit is used as the ``physical'' qubit in the cluster state. The cluster state can be extended to be arbitrarily long. Among all pairs, three pairs are measured respectively in basis $b_i$ with output results $s_i$. In-between the three pairs, there is an even number of ancillary pairs. These ancillary pairs are measured in the $X$ basis with outputs postselected to be $0$. %Each of the qubits here is represented by 2 qubits in (a).
There is a non-zero probability to get a string that satisfies the constraint in Eq.~(\ref{eq:nonlocality}).
(b) The corresponding 2-gram model.
}
\label{fig:quantum2gram}%
\end{figure}

In this section, we prove Theorem~\ref{thm:kgram}.
The separation of 2 versus $O\left(n\right)$ between the basis-enhanced $2$-gram model and classical $k$-gram models (Fig.~\ref{fig:baynet}(a)) is demonstrated through an example constructed from three-partite Bell tests of a GHZ state~\cite{mermin1990extreme,greenberger1990bell}. The GHZ state is embedded in a $n$-qubit 1D cluster state~\cite{raussendorf2001one}, such that measurement on $n-3$ qubits in the $X$ basis will produce a GHZ state (up to Pauli corrections according to the measurement results). A similar embedding was also used in Ref.~\cite{barrett2007modeling,bravyi2018quantum}. The basis-enhanced $2$-gram model is shown in Fig.~\ref{fig:quantum2gram}(a), which can be verified directly to be a BBQC, where each variable corresponds to two qubits. 

The measurement result, $b_i$, of the first qubit in the $i$th pair plays the role of choosing a measurement basis for the second qubit: $b_i = 0$ corresponds to measurement in the $X$ basis, and $b_i = 1$ corresponds to measurement in the $Y = S^\dag XS$ basis for the second qubit. All of the second qubits in each pair form a cluster state because they are connected through control-Z gates and the initial states are all $\ket{+}$. Suppose we choose three qubits to form the GHZ state and measure the remaining qubits according to Fig.~\ref{fig:quantum2gram}(b). The resulting quantum state will be a GHZ state, where the probability to get $b_i$ and $s_i$ is:
\begin{equation}
q\left(b_1s_1b_2s_2b_3s_300\ldots00\right)\propto 1+i^{b_1+b_2+b_3}(-1)^{s_1+s_2+s_3},
\end{equation}
where $b_i$ and $s_i$ denote the measurement result from the first and second qubit of the $i$th pair. When $b_1\oplus b_2\oplus b_3=0$, the strings generated by this model with non-zero probability contain and only contain $b_i$ and $s_i$ constrained by:
\begin{equation}\label{eq:nonlocality}
i^{b_1+b_2+b_3}(-1)^{s_1+s_2+s_3}=1.
\end{equation}
It can be shown that any local hidden variable theory (i.e. $s_i=s_i\left(\lambda, b_i\right)$ does not depend on $b_{j\ne i}$, where $\lambda$ is the hidden variable) cannot satisfy this equation~\cite{greenberger1990bell,bravyi2018quantum}.

We now prove that, for any $k<\left\lfloor\left(n-1\right)/2\right\rfloor$, a classical $k$-gram model cannot approximate the probability distribution $q$ generated by the above BBQC up to the error defined in Eq.~\eqref{eq:error}. We show this by reducing the classical model to a local hidden variable theory in some sense. However, classical $k$-gram models are not strictly a ``local'' theory since there is information flow from the left-most to the right-most nodes even if $k$ is a constant. There is causal influence between any pairs of nodes, i.e. a $k$-gram model can simulate the scenario that any node could communicate, though possibly only one-way, to a node on the right. In order to establish nonlocal correlations among 3 variables, communicating $2^{2\cdot(3-1)}=16$ bits of information is 
sufficient. One can cut this information flow by measuring the variables in-between. Rigorously speaking, it means one needs to show the corresponding conditional probability of the $k$-gram model is described by a local hidden variable theory:
% Our strategy is to prove that the classical model can be reduced to a local hidden variable theory in some sense. However, even the classical $k$-gram model (Fig.~\ref{fig:kgram}) is not strictly a ``local" theory even if $k$ is a constant since there is information flow from the left-most to the right-most nodes. There is causal influence between any pairs of nodes, i.e. a $k$-gram model can simulate the scenario that any node could communicate (although possibly only one-way) to a node on the right (in order to establish nonlocal correlations among 3 variables, communicating $2^{2\cdot(3-1)}=16$ bits information is enough). One way to eliminate the information flow between distant variables is to measure the variables in the middle. Rigorously speaking, we need to prove the corresponding conditional probability of the classical model is described by a local hidden variable theory by directly calculating according to the structure of $k$-gram model:
\begin{align}\label{eq:HVT}
p_C\left(b_1s_1b_2s_2b_3s_3\mid
\text{other variables are }0\text{s}\right) \nonumber \\
=\frac{f_1\left(s_1,b_1\right)f_2\left(s_2,b_2\right)f_3\left(s_3,b_3\right)}{\sum_{s_i,b_i}f_1\left(s_1,b_1\right)f_2\left(s_2,b_2\right)f_3\left(s_3,b_3\right)},
\end{align}
% \begin{widetext}
% \begin{equation}\label{eq:HVT}
% p_C\left(b_1s_1b_2s_2b_3s_3\mid
% \text{all other variables are }0\text{s}\right)=\frac{f\left(s_1,b_1\right)f\left(s_2,b_2\right)f\left(s_3,b_3\right)}{\sum_{s_i,b_i}f\left(s_1,b_1\right)f\left(s_2,b_2\right)f\left(s_3,b_3\right)},
% \end{equation}
% \end{widetext}
where $f_i\left(s_i,b_i\right)$ is the product of the terms $p\left(b_{l+k-1}s_{l+k-1}\mid b_{l}s_{l}\ldots b_{l+k-2}s_{l+k-2}\right)$ involving $s_i,b_i$ while setting other variables to be 0. Because the three variables are chosen to be further than $2k$ apart, each product only involves one variable. We can thus normalize $f_i\left(s_i,b_i\right)$ to be $p\left(s_i\mid b_i,\lambda\right)$ where $\lambda$ is determined by the measurement basis and results, as well as each term in the $k$-gram model, but $\lambda$ does not depend on $b_{j\ne i}$. This shows that Eq.~\eqref{eq:HVT} can be described by a local hidden variable theory and thus completes the proof of Theorem~\ref{thm:kgram}.

We note that the 2 vs.\ $O\left(n\right)$ separation still holds under the error in Eq.~\eqref{eq:error}, implying a separation also under the KL divergence. The circuit in Fig.~\ref{fig:quantum2gram} is essentially the same as the one used in Ref.~\cite{bravyi2018quantum} other than the boundary conditions. However, their result cannot be applied directly here since the $k$-gram model is not a constant depth classical probabilistic circuit. Concretely, $k$ does not correspond to the circuit depth and the ``light-cone'' scales with the system size even when $k$ is small. 
On the other hand, hidden Markov models with bond dimension 6 could simulate this basis-enhanced 2-gram model. We give an explicit construction in Appendix~\ref{app:hmm}.

%Because the three chosen variables are distant (larger than $2k$), each product only involves one variable (a variable is determined by both $s$ and $b$). We can normalize $f$ to be $p\left(s_i\mid b_i,\lambda\right)$ where $\lambda$ is determined by the measurement basis and results, and each term in the $k$-gram model. So it cannot produce the language such that when all the variables except $s_i,b_i$ with $i=1,2,3$ are 0s and $b_1\oplus b_2\oplus b_3=0$, $s_i,b_i$ satisfy Eq.~\eqref{eq:nonlocality}. This completes the proof of Theorem \ref{thm:kgram}.

% \begin{theorem}[$k$-gram models and quantum non-locality]\label{thm:kgram}

\subsection{Hidden Markov Models and Quantum Contexuality}\label{sec:hmm}
We now study basis-enhanced HMMs (the translation form in Fig.~\ref{fig:baynet}(b)) in the context of translation problems. We find that any classical HMM requires $D^{\Omega\left(\log D\right)}$ hidden variables in order to approximate a basis-enhanced HMM with $D$ hidden variables, under the error model of Eq.~\eqref{eq:error}. The separation originates from quantum contextuality---in particular, our proof is constructed from the Mermin--Peres Magic square~\cite{mermin1993hidden,peres1990incompatible}.

Our approach to the result is as follows. First, we discuss hidden variable theories---more precisely, \emph{ontological} theories~\cite{kochen1975problem,harrigan2007representing}---and show that they are equivalent to classical hidden Markov models. Then, we give a lower bound on the number of ontological states needed to simulate Pauli measurements on stabilizer states using the Mermin--Peres magic square~\cite{peres1990incompatible,mermin1990simple,aravind2002simple}. Finally, we discuss how basis-enhanced 2-gram models can efficiently simulate Pauli measurements on stabilizer states, proving our result.

\subsubsection{Ontological Theories and Hidden Markov Models}\label{sec:otethmm}

First, we give a description of hidden variable theories (more precisely, \emph{ontological} theories)~\cite{kochen1975problem,harrigan2007representing} in terms of hidden Markov models. An ontological theory is characterized at any moment by a state variable $\lambda_i\in\left\{\lambda_1,\lambda_2, \cdots, \lambda_V \right\}$, which we assume completely determines the resulting distribution of the measurement outcomes of various observables. In particular, the model assumes a quantum state $\ket{\psi}$ is encoded as a probability distribution over hidden variables as $p_{\ket{\psi}}(\lambda_i)$, where $\sum_i p_{\ket{\psi}}(\lambda_i) = 1$, and the measurement output from measuring an observable $\hat{O}$ is described as 
\begin{equation}\label{eq:ontological}
p_{\ket{\psi}}(y_i \mid \hat{O}) = \sum_i p(y_i \mid \lambda_i, \hat{O}) p_{\ket{\psi}}(\lambda_i),
\end{equation}
where $p_{\ket{\psi}}(y_i \mid \hat{O})$ is the quantum mechanical measurement output probability for output $y_i$ and $p(y_i \mid \lambda_i, \hat{O})$ is an indicator function independent of the quantum state $\ket{\psi}$. Below, we use the following notations as illustrated in Fig.~\ref{fig:new_HMM}(a). After a measurement of an observable $\hat{O}_{x_i}$ from some restricted set of observables $\{\hat{O}_{x_i}\}$ labeled by $x_i$, and upon obtaining a measurement result $y_i$ with probability $p(y_i\mid\lambda_i,x_i)$, there is also generally a transition probability to another state $\lambda_j$ with probability $\Gamma_M\left(\lambda_j\mid\lambda_i\right)$, where $M=\left(x_i,y_i\right)$ characterizes the measurement and its result.

\begin{figure}
\centering
\includegraphics[width=1\linewidth]{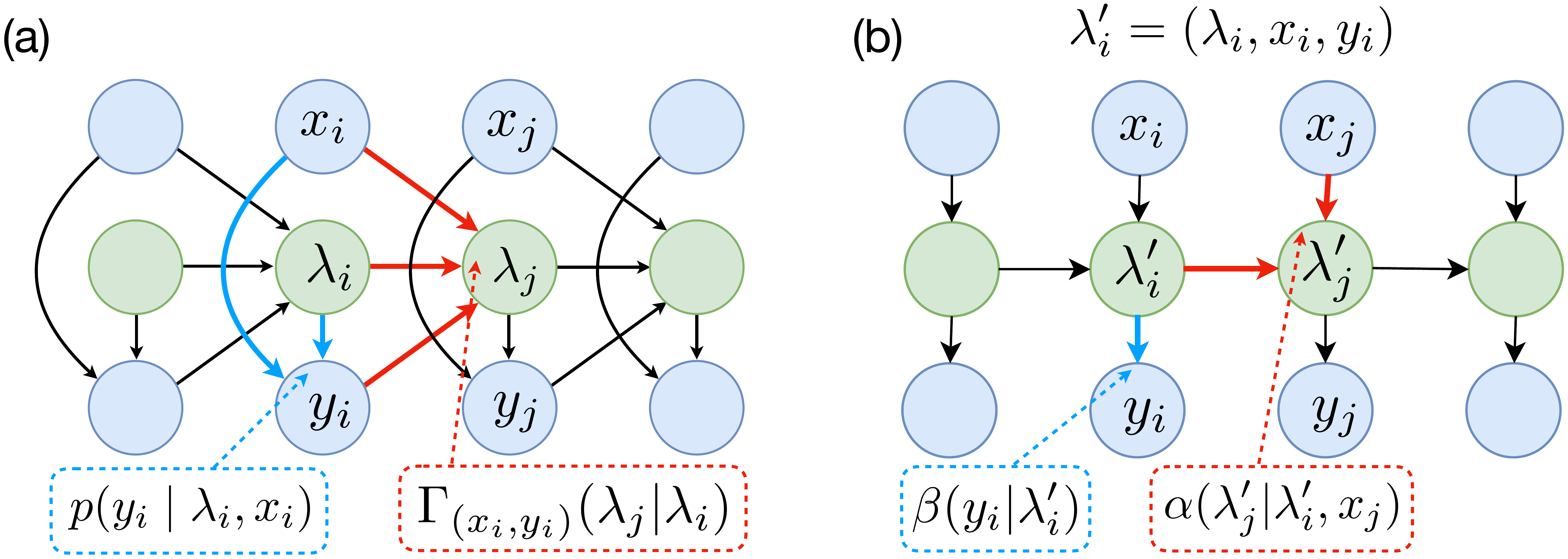}\caption{
(a) The ontological model in terms of Bayesian networks. (b) The standard hidden Markov model.
}
\label{fig:new_HMM}%
\end{figure}

As discussed in Sec.~\ref{sec:bn}, an HMM used for translation problems is a Bayesian network characterized at any moment by a hidden state $\lambda_i^\prime$, with some input $x_i$ and transitions to a new internal state $\lambda_j^\prime$ with probability $\alpha\left(\lambda_j^\prime\mid\lambda_i^\prime,x_j\right)$, and some probability $\beta\left(y_i\mid\lambda_i^\prime\right)$ emitting a symbol $y_i$ (see Fig.~\ref{fig:new_HMM}(b)). We note that if we set $\lambda_i^\prime=(\lambda_i,x_i,y_i)$,
this is identical to an ontological theory in the above paragraph.

% , with
% \begin{equation}
%     \sum_{M=(x,y)}\Gamma_M\left(\lambda^\prime\mid\lambda\right)\beta\left(y\mid\lambda,x\right)=\alpha\left(\lambda^\prime\mid\lambda\right).
% \end{equation}\todo{confused}
% Note that the form of this transition probability is exactly the form of a transition probability in a hidden Markov model, from internal state $\lambda$ to $\lambda^\prime$.
% Correspondingly, $\Gamma_{M}\left(\lambda^\prime\mid\lambda\right)=0$ for all $\lambda^\prime$ implies that measuring $\hat{O}_x$ on the state $\lambda$ will never yield the measurement result $y$. We will see that under the error model of Eq.~\eqref{eq:error2}, this is the condition we are interested in.

\subsubsection{Ontological Theories Representing Quantum States}\label{sec:Ontological}

Let us now consider how ontological theories can simulate measurements on a quantum system. We follow Ref.~\cite{karanjai2018contextuality} on the discussion of ontological models, only keeping concepts that are relevant for our goal. We use Mermin--Peres magic squares explicitly to demonstrate the advantage of quantum contextuality.

A naive way to simulate a quantum system subject to sequential measurements is by recording each quantum state $\ket{\psi_i}$ that the system could generate as its own state variable $\lambda_i$. Though this encodes all information in a quantum theory, there is a large overhead in terms of the number of internal states $\lambda_i$ needed, depending on which classes of circuits are modeled. We thus consider encodings that allow an internal state $\lambda_i$ to be shared by different quantum states $\ket{\psi}$. In this case, each quantum state is encoded as a probability distribution over $\left\{\lambda_1,\lambda_2,\cdots, \lambda_V\right\}$. Since we consider the error model in Eq.~\eqref{eq:error}, we only need to consider if a measurement probability is zero or not, while the precise values of the probabilities are not important. Thus, a quantum state can be associated with 
%, perhaps with some overlap depending on the measurement history of the system being and the probabilistic internal transitions of the ontological theory.
% For example, if we only have the following three states and observable involved: $\{\ket0,\ket1,\ket+,\ket-\}$ and $\{Z\}$, only two hidden variables are enough to generate the same measurement outcomes by assigning $\ket0\sim\{\lambda_1\},\ket1\sim\{\lambda_2\}$ and $\ket\pm\sim\{\lambda_1,\lambda_2\}$. Here we use a subset to correspond to a quantum state because we do not care the value of non-zero probability (in general case, a quantum state corresponds to a probability distribution of hidden variables). Later we will use
%In particular, we will be interested in 
a \emph{support}% (a subset of $\{\lambda_1,\lambda_2,\cdots\}$)
\begin{equation}
    \operatorname{supp}\left(\ket{\psi}\right) = \left\{ \lambda_i \mid p_{\ket{\psi}}(\lambda_i) \neq 0 \right\},
    \label{eq:supp}
\end{equation}
which means the subset of internal states that the ontological theory could be in when representing the quantum state $\ket{\psi}$. %, over all possible measurement histories.

As illustrated in Fig.~\ref{fig:baynet}(b), we interpret the translation form of the HMM from a dynamical point of view. The state of the HMM is encoded as a probability distribution $p_{\ket{\psi_t}}$ over the hidden states of the HMM at the $t$-th time step. The quantum state at time $t$, $\ket{\psi_t}$, depends on all the previous measurement outcomes, $M_t^{\text{past}} \equiv \left\{M_1,M_2,\cdots,M_{t-1}\right\}$. Thus, in order to faithfully simulate the quantum process, the HMM should have enough memory about all previous measurement bases and outcomes to predict future behavior. The number of hidden states corresponds to the memory of the system. We could define the union of all the states at time $t$ resulted from different measurement outcomes, $\left\{\ket{\psi_t (M_t^{\text{past}})}\right\}$, but there is ambiguity in setting the weights for different measurement histories. However, under the error model in Eq.~\eqref{eq:error}, one only needs to care about whether the probabilities are nonzero or not. Each measurement history can be associated with a support of hidden variables, i.e. $\operatorname{supp}\left(\ket{\psi_t (M_t^{\text{past}})}\right)$, and the union over different histories can be defined as the union of the support spaces. We emphasize that this is well-defined because there is no interference in the hidden Markov model, i.e.\ summation of different histories cannot be cancelled.

% Strictly speaking, we cannot associate a probability distribution over $\{\lambda_1,\lambda_2,\cdots, \lambda_V\}$ to a quantum state directly. We could define a probability distribution for a history of previous measurement (sequence of $M$s). In quantum mechanics, a history defines a quantum state. However, there might be different measurement histories to reach $\ket\psi$. % Starting from some certain hidden variable, after a measurement history (a sequence of $M$), there is a definite probability distribution in the ontological theory and there is a definite quantum state in quantum theory. 
%  In this case, we may try to define %$supp(\ket\psi)$ as 
%  the union of all the distributions of measurement histories which define the state $\ket{\psi}$. But it is hard to get well-defined probability distribution for a quantum state since it's ambiguous to set the weights of different histories. Fortunately, under the error model in Eq.~(\ref{eq:error}), each history is associated with a support of hidden variables then the union over different histories becomes unique. Notice that this is well-defined because there is no interference in hidden Markov model, i.e., summation of different histories can not be cancelled.

Naively, one might believe that it is possible to encode $2^V-1$ quantum states using only $V$ ontological states, since a set with $V$ elements has $2^V-1$ nontrivial subsets. However, in order for the ontological theory to make the same predictions as the quantum theory, there are restrictions on which subsets of hidden variables are used to label quantum states. For instance, if two states $\ket{\psi_1}$ and $\ket{\psi_2}$ are eigenstates with different eigenvalues of an allowed observable $\hat{O}$, then we must have $\operatorname{supp} \left(\ket{\psi_1}\right)\cap \, \operatorname{supp} \left(\ket{\psi_2}\right)=\emptyset$.
%there is a so-called \emph{separating measurement} in $\mathcal{O}$ between two states $\ket{\psi_1}$ and $\ket{\psi_2}$---i.e. an observable $\hat{O}$ such that $\bra{\psi_1}\hat{O}\ket{\psi_1}=-\bra{\psi_2}\hat{O}\ket{\psi_2}=\pm 1$---then $\operatorname{supp}\left(\ket{\psi_1}\right)\cap\operatorname{supp}\left(\ket{\psi_2}\right)=\emptyset$, 
%as otherwise their supports has overlap thus leads to a contradiction by measuring $\hat{O}$. 
As an example to illustrate this, suppose $\hat{O}\ket{\psi_1}=-\ket{\psi_1}$ and  $\hat{O}\ket{\psi_2}=\ket{\psi_2}$, and there is at least one overlapping hidden variable in the support denoted as $\lambda_i$. Let us first assume the output is always $+1$ when measuring $\lambda_i$ by $\hat{O}$, i.e. $p\left(+1 \mid \lambda_i, \hat{O}\right) = 1$; then, there is a nonzero probability for both $\ket{\psi_1}$ and $\ket{\psi_2}$ to obtain the measurement result $+1$ according to Eq.~\eqref{eq:ontological}, which contradicts the prediction from quantum measurements. Similarly, assuming the output of measuring $\lambda_i$ by $\hat{O}$ is always $-1$ or nonzero on both $\pm1$ also leads to the same contradiction.

%they could both be represented by the same hidden state $\lambda$ and give the same measurement outcome when measuring $\hat{O}$.
%Another a little bit less trivial example is that even though $\innerp{\psi_1}{\psi_2}\ne0$, if there exists an $M$ (an observable with a specific measurement result) s.t. the resulting states are orthogonal $\innerp{\phi_1}{\phi_2}=0$, then we still have $supp(\ket{\psi_1})\cap supp(\ket{\psi_2})=\emptyset$. The reason is that, suppose there is a common hidden variable $\lambda$

A less trivial example is given by considering quantum contextuality, where the intersection among the supports of several quantum states should still be empty even if there is not a pair of states that are orthogonal. We now proceed to use the Mermin--Peres square to construct such triplets of states.

\subsubsection{No Common Hidden Variables in the Mermin--Peres Magic Square}
In the following, we focus our attention on stabilizer states~\cite{gottesman1997stabilizer}. We first extend in a more formal way the discussion between memory and contextuality through the Mermin--Peres magic square example mentioned in Sec.~\ref{sec:summary}. Then, in Sec.~\ref{sec:mpbound}, we prove a lower bound on the number of hidden states required to simulate Pauli measurements on all stabilizer states. 

Given three stabilizer states 
 $\ket{\psi_1}$, $\ket{\psi_2}$, and $\ket{\psi_3}$, let $\left\{\hat{O}_{i1}\right\}$, $\left\{\hat{O}_{i2}\right\}$, and $\left\{\hat{O}_{i3}\right\}$ be their corresponding stabilizer groups. If 
\begin{equation}\label{eq:observables}
\left\{\hat{O}_{i1}\right\}\cup\left\{\hat{O}_{i2}\right\} \cup \left\{\hat{O}_{i3}\right\}
\end{equation}
contains nine observables which form a Mermin--Peres square as shown in Table~\ref{tab:square}, then we show by contradiction that it must be true that
\begin{equation}\label{eq:3states}
\operatorname{supp}\left(\ket{\psi_1}\right)\cap \operatorname{supp}\left(\ket{\psi_2}\right)\cap \operatorname{supp}\left(\ket{\psi_3}\right)=\emptyset.\end{equation}

\begin{table}[ht]
  \begin{center}
    \begin{tabular}{c|c|c|c|l} % <-- Alignments: 1st column left, 2nd middle and 3rd right, with vertical lines in between
      \cline{2-5}
     $\ket{\psi_1}$& A & a & Aa & $+1$\\
      \cline{2-5}
    $\ket{\psi_2}$&  B & b & Bb& $+1$\\
      \cline{2-5}
     $\ket{\psi_3}$& AB & ab & ABab &$+1$\\
     \cline{2-5}
     & $+1$& $+1$ & $-1$&
    \end{tabular}
  \end{center}
  \caption{A Mermin--Peres magic square. The commutation relations among these operators are given in the following equations:
  $\left[A,a\right]=\left[B,b\right]=\left[A,B\right]=\left[a,b\right]=0$ and $
\left\{A,b\right\}=\left\{a,B\right\}=0$.}\label{tab:square}
\end{table}

Concretely, we consider an example %in Sec.~\ref{sec:summary} 
to illustrate this idea and leave the general proof in Appendix~\ref{app:stab}.
Let:
\begin{align}\label{eq:eg}
\ket{\psi_1}&=\ket{00}, \ket{\psi_2}=\ket{++}, \nonumber \\ \ket{\psi_3}&= \frac{\ket{00}+\ket{01}+\ket{10}-\ket{11}}{2}, \\
A&=Z_1, a=Z_2, B=X_2, b=X_1. \nonumber
\end{align}
These stabilizer states and stabilizers form a Mermin--Peres magic square as in Table~\ref{tab:square}. Here, we show that the intersection among the supports of the above three states must be empty. Suppose instead the intersection contains a hidden variable $\lambda_i$. We consider a measurement of the stabilizer $ABab=Y_1Y_2$. Since $\ket{\psi_3}$ is its eigenstate of $ABab$ with eigenvalue $+1$, we have that for any ontological theory in state $\lambda_i$ belonging to the support of $\ket{\psi_3}$:
\begin{equation}
p\left(y=+1\mid \hat O_y=Y_1Y_2,\lambda_i\right)=1
\end{equation}
in order to make the same prediction with quantum mechanics and since $\sum_i p_{\ket{\psi_3}}(\lambda_i) = 1$. In addition, there exists a $\lambda_j$ such that
\begin{equation}
    \Gamma_{\left(Y_1Y_2,+1\right)}\left(\lambda_j\mid \lambda_i\right)>0.
\end{equation}
Since $\lambda_i$ also belongs to the supports of $\ket{\psi_1}$ and $\ket{\psi_2}$, $\lambda_j$ must also belong to the supports of:
\begin{eqnarray}
\frac{(1+Y_1Y_2)}{2}\ket{\psi_1}&\propto &\frac{1}{\sqrt{2}}\left(\ket{00}-\ket{11}\right)\\\nonumber
\frac{(1+Y_1Y_2)}{2}\ket{\psi_2}&\propto &\frac{1}{\sqrt{2}}\left(\ket{01}+\ket{10}\right),
\end{eqnarray}
which are the resulting states after measuring $Y_1Y_2$ on states $\ket{\psi_1}$ and $\ket{\psi_2}$ and getting the measurement result $+1$. However, these two states are orthogonal and thus cannot share a common $\lambda_j$ as explained in Sec.~\ref{sec:Ontological} (e.g.\ consider measuring $Z_1Z_2$ or $X_1X_2$). We thus arrive at a contradiction and the three states cannot share a common hidden variable $\lambda_i$. It is straightforward to extend this example to more general stabilizers with the same commutation relations as those in Table~\ref{tab:square}, which is detailed in Appendix~\ref{app:stab}.

\subsubsection{Bounding the Efficiency of Hidden Markov Models}\label{sec:mpbound}

We now prove a lower bound, which has been used in Ref.~\cite{karanjai2018contextuality}, for the number of hidden states needed in the HMM. Denote $S$ as the set of all the possible quantum states appearing in the quantum system we want to simulate using an ontological theory; in our case here, it is the set of all stabilizer states. Denote $s$ as a subset of $S$ such that
\begin{equation}
\bigcap_{\ket\psi\in s}\operatorname{Supp}\left(\ket{\psi}\right)\ne\emptyset.
\end{equation}
Let
\begin{equation}
m=\max_s\left\lvert s\right\rvert.
\end{equation}
Then, the number of state variables, $V$, needed in an ontological theory in order to simulate the quantum system is not smaller than $\left\lvert S\right\rvert/m$. The reason is illustrated in Fig.~\ref{fig:hypergraph}.

\begin{figure}[tbp]
\includegraphics[width=0.5\textwidth]{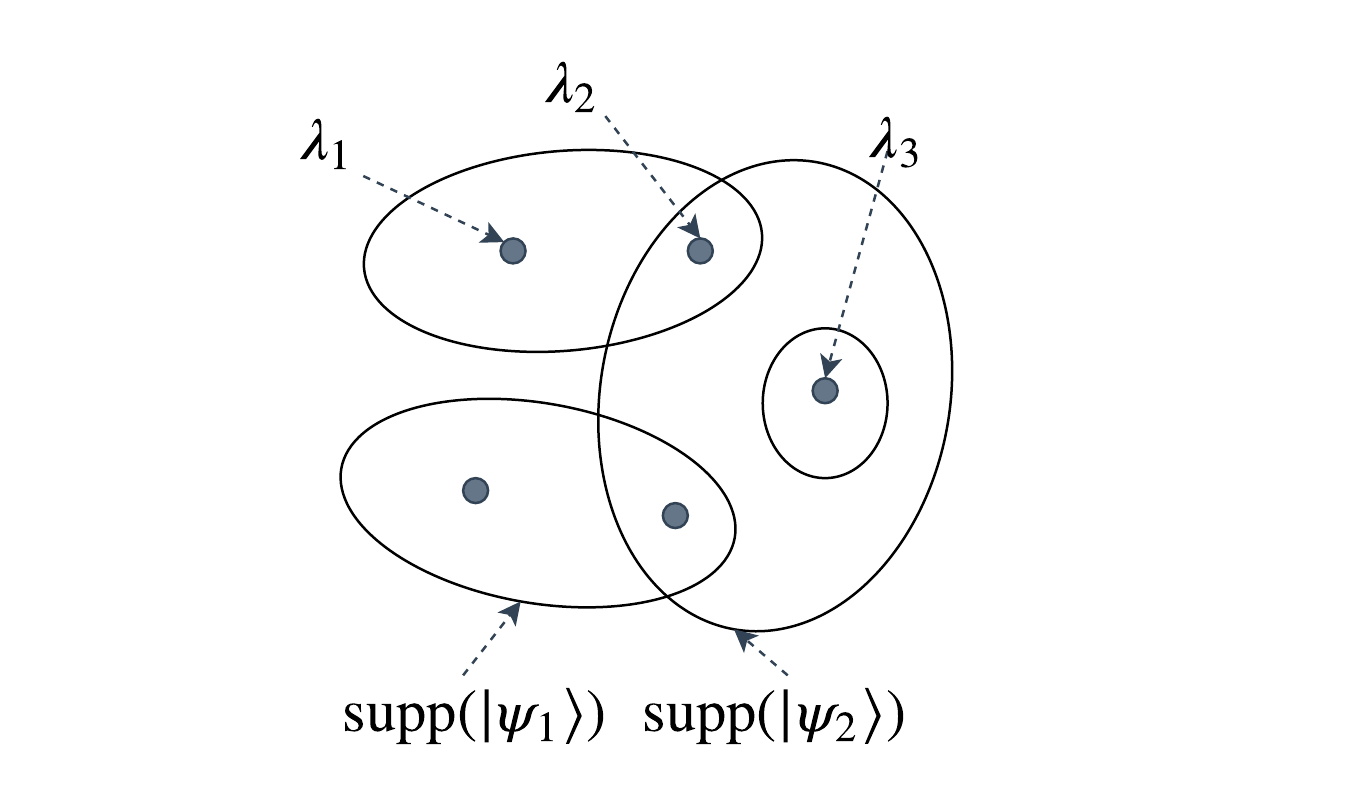}
\caption{Each circle corresponds to a support of a quantum state; there are thus $\left\lvert S\right\rvert$ circles. Each dot corresponds to a hidden variable. Denote the number of dots as $V$, and the number of circles that a dot belongs to is at most $m$. Imagine that we eliminate a dot and its associated circles each time. After $V$ steps, all of the circles are eliminated. Since each time we can eliminate at most $m$ circles, we have $Vm\ge\left\lvert S\right\rvert$.
}
\label{fig:hypergraph}
\end{figure}

The following two lemmas show that by allowing for all stabilizer states and all Pauli measurements, we have large $\left\lvert S\right\rvert$ and relatively small $m$. We upper bound $m$ by using contextuality and proving the existence of the Mermin--Peres square.
\begin{lemma}[Proposition 1 in Ref.~\cite{aaronson2004improved}] \label{lem:NoStabilizers}
The total number of stabilizer states is
\begin{equation}
\left|S\right|=2^{n^2/2+o(n^2)},
\end{equation}
where $n$ is the number of qubits.
\end{lemma}

\begin{lemma}\label{lem:square}
For any subset of stabilizer states $s$ such that $\left|s\right|>2^{n^2/4+7n/2}$, there exist three states such that some of their stabilizers forming a Mermin-Peres magic square as shown in Table~\ref{tab:square}. Therefore, $m\leq 2^{n^2/4 + 7n/2}$.
\end{lemma}
The idea of the proof is by showing the following: (1) up to classical Clifford circuits (those only composed of CNOT and X), any stabilizer state can be written as a tensor product of a uniform superposition state of $k$ qubits (with possibly non-trivial phases) and a product state of $n-k$ over computational basis; (2) for any $2^{n^2/4+3n/2}\cdot 4^n \cdot 2$ states, we can always find one computational basis product state and $4^n \cdot 2$ other states that are superposition over the \emph{same} $k$ qubits up to a Clifford circuit; (3) we can always find 2 standard graph states from these $4^n  \cdot 2$ states up to single-qubit phase gates (which do not change the computational basis product state); (4) finally, the stabilizers of the computational basis product state and the 2 graph states can always form a Mermin-Peres square. Details of the proof are in Appendix~\ref{app:stab}. From both lemmas, we have a lower bound for $V$:
\begin{equation}
    V \geq \frac{\left\lvert S\right\rvert}{m}\geq 2^{\frac{n^2}{4}-O(n)}.
\end{equation}
Following the discussion on the equivalence of ontological theories and HMMs in Sec.~\ref{sec:otethmm}, the lower bound on $V$ shows that any ``translation form'' HMM simulating Pauli measurements on stabilizer states on $n$ qubits requires at least $2^{\Omega\left(n^2\right)}$ hidden states. In the following section, we show that a basis-enhanced 2-gram model, which is a special case of the ``translation form'' HMM as in Fig.~\ref{fig:baynet}(b) (see Appendix~\ref{app:equivalence}), only needs $2^{O\left(n\right)}$ internal states.
% Choosing $k=2^n$ as in Fig.~\ref{fig:Clifford}, the number of hidden variables required in the hidden Markov model is at least $k^{\Omega\left(k\right)}$.

We have two remarks on the above proof. First, the idea does not need to be restricted to stabilizer states. Here, we use stabilizer states for the simplicity of illustrating the key idea, i.e. making use of the simplest example of contexuality, the Mermin--Peres square. Second, the above argument also works for HMMs without translational invariance. One can remove the requirement of translational invarance by just relabeling $\Gamma^{\left(t\right)}$ and $\operatorname{Supp}^{\left(t\right)}$ for the $t$-th time step.

\begin{figure}[tbp]
\includegraphics[width=0.45\textwidth]{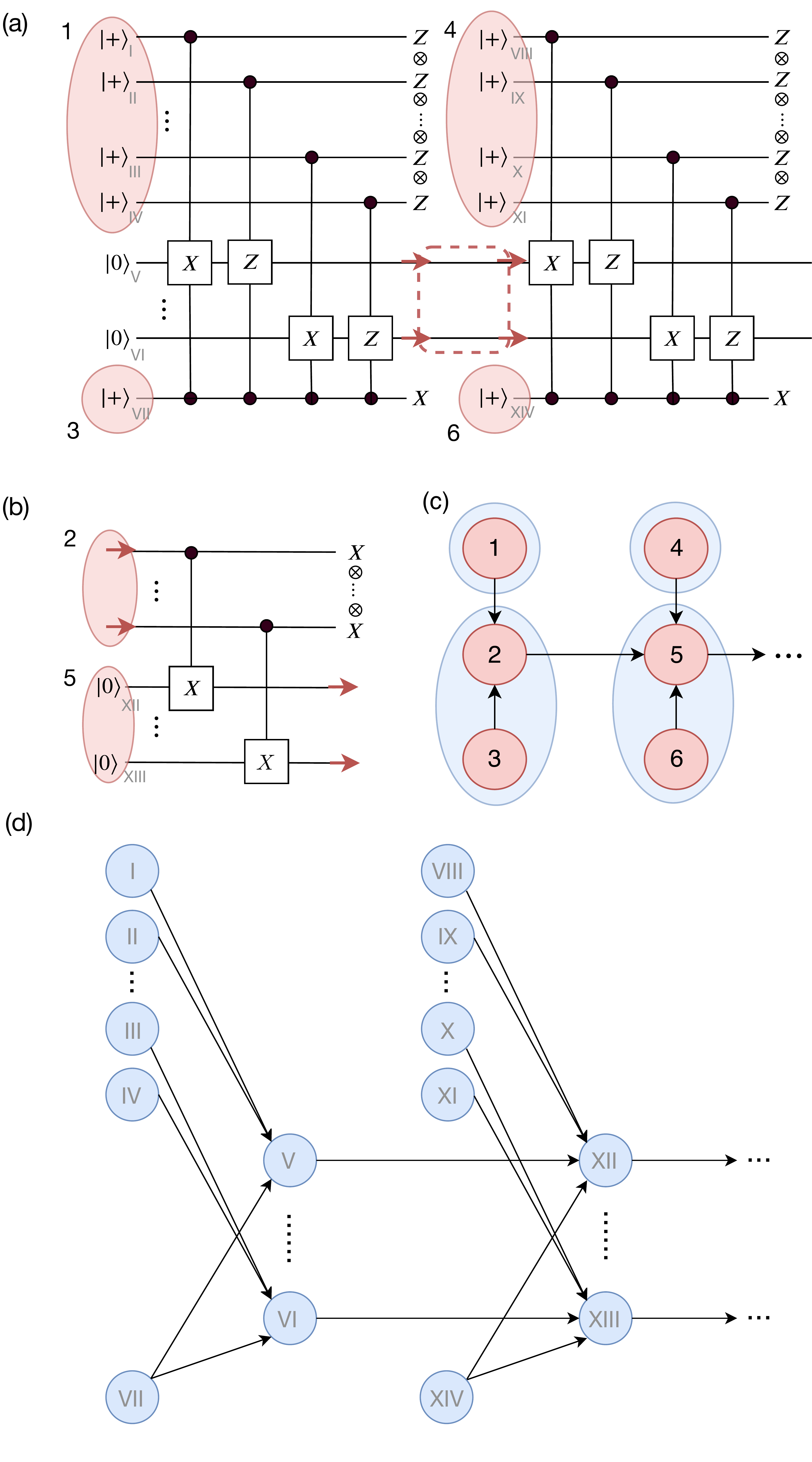}
\caption{(a) Quantum circuits for generating all of the possible stabilizer states by measurements. The circles 1 and 4 are superpositions of $2n$-bit strings with uniform weights, of which measurement results in the $Z$ basis determine which stabilizers are to be measured for the qubits initialized as $\ket{0}^{\otimes n}$ (those qubits in the middle). The measurement results in the $X$ basis of circles 3 and 6 give the outcome of measuring the corresponding Pauli determined by the measurement result of circles 1 and 4, respectively. All of the outcomes of circles $1,3,4,6,\cdots$ will make the qubits in the middle be any possible stabilizer state. (b) Teleportation gadget for the part in the red dashed box of (a). The stabilizer to be measured in the next time step should be determined by measurement results of both 2 and 5 because the measurement result of 2 will sometimes produce some Pauli corrections.
(c) The corresponding Bayesian network circuit of (a) after using (b). All of the red circles correspond to circles in (a) and (b) with the same number labels. We group $\left\{1\right\}$, $\left\{2,3\right\}$ and $\left\{4\right\}$, $\left\{5,6\right\}$ as blue circles to form a 2-gram model with input, which is a special case of the translation form of the HMM given in Fig.~\ref{fig:baynet}(b).
(d) 2D arrangement of the Bayesian network for the circuit shown in (a). The correspondence between quantum circuits and circles are labeled by Roman numerals.
}
\label{fig:Clifford}
\end{figure}

\subsubsection{A Basis-Enhanced 2-gram Model from Stabilizer States}

Here, we construct a basis-enhanced 2-gram model as shown in Fig.~\ref{fig:Clifford} that simulates Pauli measurements on stabilizer states using $O\left(n\right)$ qubits---the underlying 2-gram model, therefore, has $D = 2^{O\left(n\right)}$ internal states. The input and output of the HMM correspond to the choices of different Pauli measurements and the measurement results respectively, and the teleportation gadgets connect successive nodes in the 2-gram model. Together with the result in the previous section~\ref{sec:mpbound} on the lower bound for the number of hidden variables required in a classical HMM to simulate the quantum process, $V \geq 2^{\Omega\left(n^2\right)}=D^{\Omega\left(\log D\right)}$, we have thus completed the proof of Theorem~\ref{thm:hmm}. Furthermore, instead of regarding it as a 1D model with large bond dimension, we can also view it as a 2D model:
in Fig.~\ref{fig:Clifford}(a), if each qubit is regarded as a node in the Bayesian network, the corresponding directed graph is shown in Fig.~\ref{fig:Clifford}(d).

Before concluding this comparison between HMMs and basis-enhanced HMMs, we note that although this basis-enhanced circuit is constructed from stabilizer states, which can be efficiently simulated classically, this does not imply that a quantum computer is not useful in this case. In particular, one can consider continuous Pauli rotations instead of Clifford gates, which will be required in practice in order to train the model using a method such as gradient descent. Such quantum training algorithms for these models are analyzed in Appendix~\ref{app:quantum_algorithm}. In general, such algorithms cannot be simulated efficiently classically.

% The quantum bigram model shows that there is a quantum advantage in expressive power so choosing it as the generative model is reasonable. However, the example shown here could be only the final generative model after training. Because there is no continuous parameters in stabilizer states and training is usually based on gradient descent which requires continuous parameters. In the training process, using quantum computer is required. Besides, unlike training models on general graph of which efficient computation is not guaranteed (see Appendix \ref{app:quantum_algorithm}), we can extend the circuit in Fig.\ref{fig:Clifford}(a) by generalizing control-control-Pauli to general control-control-unitary. This extension could be trained efficiently by quantum computer but classical algorithm could not.

\section{Numerical Tests on Real World Data}\label{sec:numerics}%applications on real world data

% \subsection{Loss function}

\begin{figure*}[t!]
  \begin{center}
    \includegraphics[width=\linewidth]{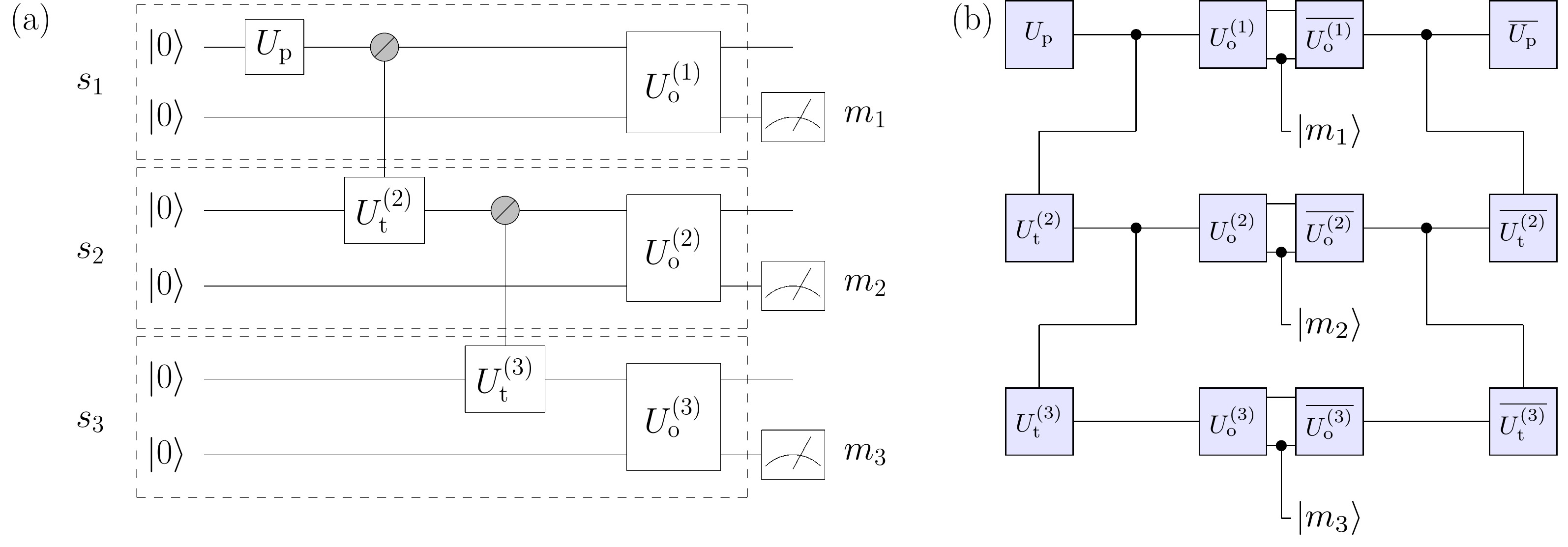}
    \caption{(a) A BBQC interpretation of a three-node $2$-gram model, with states $s_1, s_2, s_3$. An initial unitary $U_{\textrm{p}}$ constructs the prior, and uniformly controlled unitaries $CU_{\textrm{t}}$ encode transitions between $2$-gram states. Each measurement is in some standard basis encoding of the $M$-dimensional output space. To directly compare with classical HMMs, we only consider transfer unitaries on a $k$-dimensional subspace of each state (top wire), and consider the distribution on a $M$ subspace of each state (bottom wire). (b) A tensor network representation of the same quantum circuit.\label{fig:q2gram_circuit}}
  \end{center}
\end{figure*}

In the previous section, we have proven theoretically that the quantum models we consider have more expressive power than the corresponding classical models. The sources of the quantum advantages are quantum nonlocality and contextuality. In this section, we numerically test that the quantum models do indeed have better performance in practice. These numerical results primarily serve two purposes. First, it demonstrates that the quantum models actually have an advantage on real world data. Second, it shows that the quantum advantages are robust to more practical error models beyond the one used theoretically as in Eq.~\eqref{eq:error}. 

Concretely, we specialize to classical HMMs and the quantum extension of $2$-gram models introduced in Sec.~\ref{sec:hmm}. As in most generative modeling tasks, the quantity of interest to evaluate the performance of the parameterized model $p_{\textrm{model}}$ given a data set $p_{\textrm{data}}$ is the forward KL divergence
\begin{equation}
    D_{\textrm{KL}}\left(p_{\textrm{data}}\mid\mid p_{\textrm{model}}\right)=\!\!\!\!\!\! \sum\limits_{\bm{m}\in{\left[1,\ldots,M\right]}^n} \!\!\!\! p_{\textrm{data}}\left(\bm{m}\right)\log\frac{p_{\textrm{data}}\left(\bm{m}\right)}{p_{\textrm{model}}\left(\bm{m}\right)};
    \label{eq:forward_kl}
\end{equation}
as consistent with our convention used in previous sections, we let $M$ denote the dimensionality of a given visible node in our model, and let $n$ be the number of visible nodes in the model. Since summing over the exponential number of terms in Eq.~\eqref{eq:forward_kl} is intractable in practice, we use the stochastic estimate of the KL divergence given in Eq.~\eqref{eq:stoch_kl_div}.

\subsection{Simulation of Basis-Enhanced $2$-gram Models}\label{sec:sim}

We now specialize to (translationally invariant) classical hidden Markov models and basis-enhanced $2$-gram models, both of which were introduced in Sec.~\ref{sec:hmm}. Though in Sec.~\ref{sec:hmm} we considered a specific translation task for the sake of our analysis, here we consider general basis-enhanced $2$-gram models, with the parameters trained to represent some given data set. The general structure of the model we consider is given in Fig.~\ref{fig:q2gram_circuit}.

Though basis-enhanced $2$-gram models cannot directly be interpreted as classical Bayesian networks, they are still classically simulable using tensor networks when they have low bond dimension~\cite{PhysRevX.8.031012,NIPS2019_8429}, making them a natural choice for numerical tests of our analysis (see Fig.~\ref{fig:q2gram_circuit}(b)). In particular, the direction of steepest descent of~\eqref{eq:stoch_kl_div} when varying a particular tensor $U$ is given by its negative gradient with respect to the conjugate of the parameters~\cite{IMM2012-03274}; that is, the direction of steepest ascent with respect to $U$ takes the form:
\begin{equation}
    \pdv{\tilde{D}_{\text{KL}}}{\overline{U}}=\frac{2\partial_{\overline{U}}Z}{Z}-\frac{2}{K}\sum\limits_{\bm{m}\in\textrm{data}}\frac{\partial_{\overline{U}}S_{\bm{U}}\left(\bm{m}\right)}{S_{\bm{U}}\left(\bm{m}\right)},
    \label{eq:derivative_of_loss}
\end{equation}
where $S$ is the unnormalized probability distribution given in Fig.~\ref{fig:q2gram_circuit}(b) and $Z$ is its normalization. When we perform the Riemannian descent algorithm described in Sec.~\ref{sec:opt}, we optimize on the manifold of unitary matrices and thus $Z=1$. As we maintain translational invariance in our model, the total derivative with respect to some parameter $\overline{U}$ is given by the sum of the variation over all equivalent tensors:
\begin{equation}
    \dv{\tilde{D}_{\text{KL}}}{\overline{U}}=\sum_j\pdv{\tilde{D}_{\text{KL}}}{\overline{U^{(j)}}}.
\end{equation}
For completeness, we give examples of the tensor network representations of $\partial_{\overline{U}}S_{\bm{U}}\left(\bm{m}\right)$ and $\partial_{\overline{U}}Z$ in Appendix~\ref{app:avg_perf}. Since within one training mini-batch many of the same tensors are contracted, in practice we precompute intermediate tensor contraction results for each mini-batch. For a basis-enhanced $2$-gram model with bond dimension $h$, the classical runtime is $O\left(nh^3 M\right)$ for computing the gradient with respect to the unitaries in the model. For comparison, a classical HMM trained using the Baum--Welch algorithm~\cite{baum1970} takes time $O\left(nh\left(h+M\right)\right)$ per training iteration.

\begin{figure*}[ht]
  \begin{center}
    \includegraphics[scale=0.275]{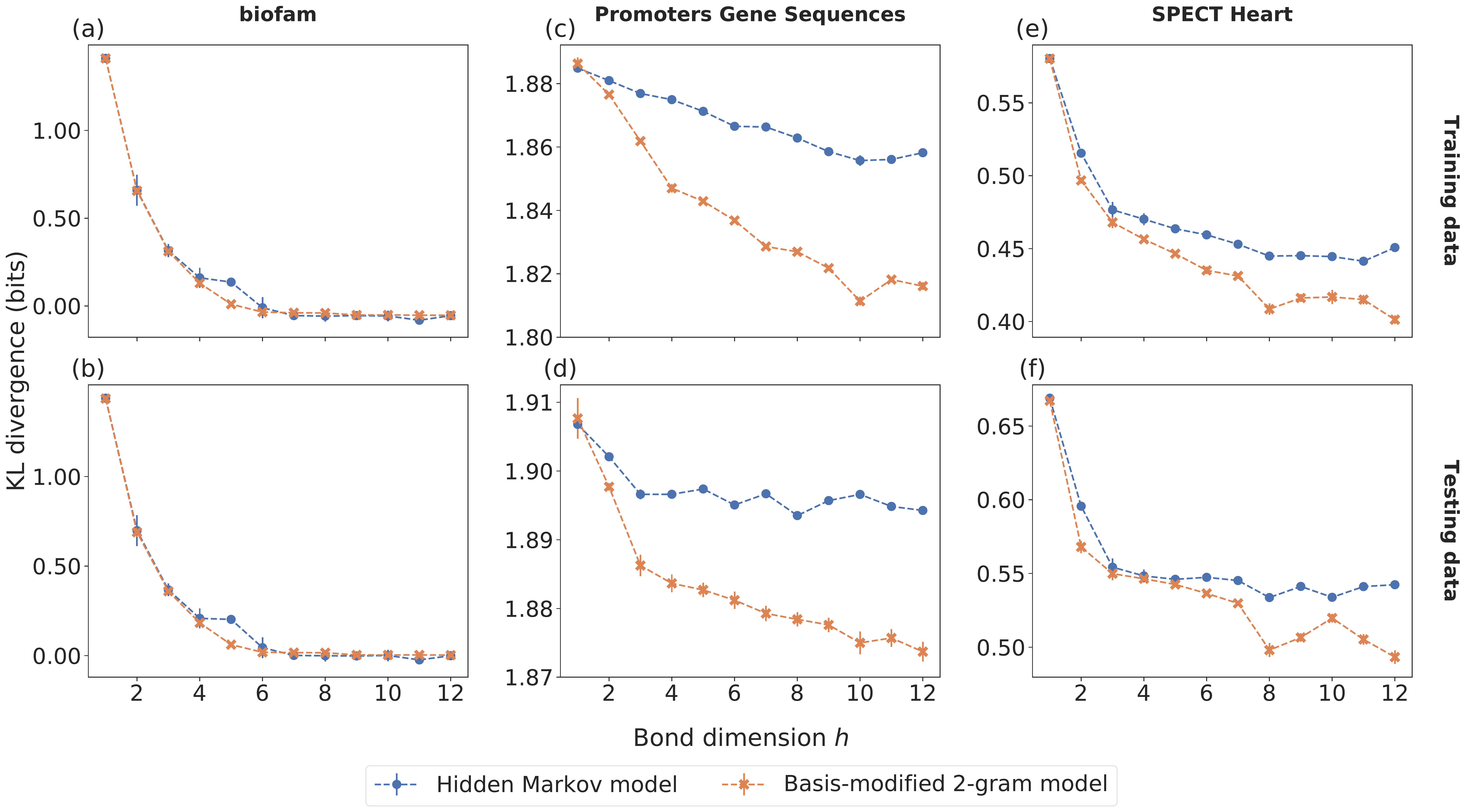}
    \caption{Shown are the best performances over ten trials of the classical HMM (blue circles) and the basis-enhanced $2$-gram model (orange crosses) on the biofam ((a) and (b)), Promoter Gene Sequences ((c) and (d)), and SPECT Heart ((e) and (f)) data sets. The first row plots the performance on the training data, and the second the performance on withheld testing data. The basis-enhanced $2$-gram models have better performance than classical HMMs on the SPECT Heart and Promoter Gene Sequence data sets. Error bars denote one standard error of the mean over ten trials. Dashed lines are to aid the eye.\label{fig:performance_plot}}
  \end{center}
\end{figure*}

\subsection{Model Training}\label{sec:opt}

In general, training Bayesian networks beyond tree graphs is hard~\cite{chickering1996learning,dagum1993approximating}. We note that there exist many heuristic and approximate algorithms that work well in practice for training classical Bayesian networks~\cite{koller2009probabilistic}, and we consider similar heuristics for BBQCs here, as described in more detail in Appendix~\ref{app:quantum_algorithm}.

Since we focus on translationally invariant HMMs here, we can use the Baum--Welch algorithm to efficiently train the classical model~\cite{baum1970}. Furthermore, as discussed in Sec.~\ref{sec:sim}, computing the gradient of the loss function with respect to the parameters in the basis-enhanced $2$-gram model is classically efficient using tensor networks for small bond dimension. However, naively performing gradient descent on the parameters of the model would generally violate unitarity constraints in the underlying quantum circuit model. Therefore, to optimize the unitaries used in the construction of the quantum model, we perform a variant of the Riemannian gradient descent algorithm introduced in Ref.~\cite{abrudan2008steepest}.

Normally, in the gradient descent of some loss function $\mathcal{L}\left(\left\{A_i\right\}\right)$ for complex-valued matrices $A_i$, one iteratively estimates the optimal $A_i$ through the update rule~\cite{IMM2012-03274}:
\begin{equation}
    A_i\to A_i-\alpha\pdv{\mathcal{L}}{\overline{A_i}},
\end{equation}
where $\alpha$ is the learning rate. In practice, keeping a moving average of previous gradient estimates smooths out stochastic fluctuations in estimates of $\pdv{\mathcal{L}}{\overline{A_i}}$; thus, we consider the momentum-based update rule~\cite{rumelhart1986learning}:
\begin{align}
    v_{A_i}&\to\beta v_{A_i}+\alpha\pdv{\mathcal{L}}{\overline{A_i}},\\
    A_i&\to A_i-v_{A_i}.\label{eq:momentum_update}
\end{align}

For unitary $A_i$, however---as in the case of quantum circuits---this procedure will generally yield nonunitary $A_i$. Therefore, we analytically calculate the direction of steepest descent in unitary space in terms of $\pdv{\mathcal{L}}{\overline{A_i}}$, and perform parallel transport in that direction~\cite{abrudan2008steepest}. This leads to the update rule for a unitary matrix $U$:
\begin{equation}
    U\to\exp\left(-\alpha\left(U\left(\pdv{\mathcal{L}}{\overline{U}}\right)^\dag-\pdv{\mathcal{L}}{\overline{U}}U^\dag\right)\right)U.
\end{equation}
We modify the method in Ref.~\cite{abrudan2008steepest} slightly to allow for the momentum update rule of Eq.~\eqref{eq:momentum_update}; namely, we use the update rule:
\begin{equation}
    U\to\exp\left(-\left(Uv_U^\dag-v_U U^\dag\right)\right)U.
\end{equation}

\begin{figure}[ht]
  \begin{center}
    \includegraphics[scale=0.45]{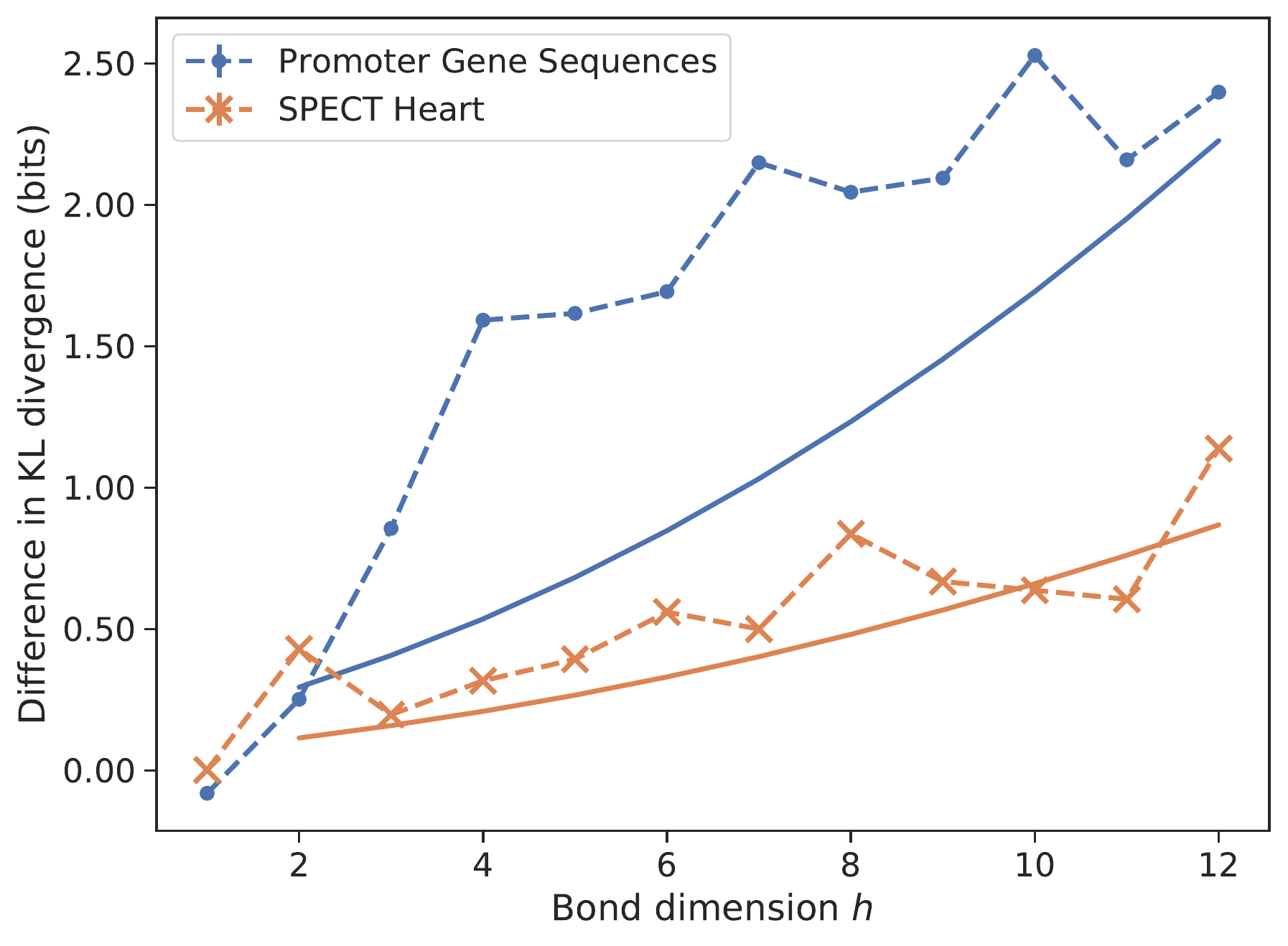}
    \caption{Plotted with dashed lines are the improvements in KL divergence between the best performing basis-enhanced $2$-gram model and the best performing classical hidden Markov model for each tested model size on the Promoter Gene Sequences and SPECT Heart data sets. The solid lines show the improvement needed to reject the null hypothesis in a likelihood-ratio test with $5\sigma$ confidence (see Sec.~\ref{sec:model_comparison}). Wherever the dashed lines are above their corresponding solid lines, the null hypothesis was rejected with $>5\sigma$ confidence.\label{fig:p_value_plot}}
  \end{center}
\end{figure}

\subsection{Model Comparison on Data Sets}\label{sec:model_comparison}

We test the performance of our implemented quantum extension of a $2$-gram model on three data sets: the biofam (sequence length $n=16$, output dimensionality $M=8$)~\cite{Mueller2007,ritschard2013exploratory}, Promoter Gene Sequences ($n=57$, $M=4$)~\cite{Dua:2019}, and SPECT Heart ($n=23$, $M=2$)~\cite{Dua:2019} data sets. For all of our simulations, we use $\beta=0.5$. For the biofam data set~\cite{Mueller2007,ritschard2013exploratory}, we used $\alpha=10^{-3}$, and for the Promoter Gene Sequences and SPECT Heart data sets~\cite{Dua:2019}, we used $\alpha=10^{-2}$. For the biofam data set, we trained for $75$ epochs, and for the Promoter Gene Sequences and SPECT Heart data set, we trained for $150$ epochs. For all data sets, we estimated the gradient over a mini-batch size of $8$ training samples. The biofam data set tracks the family life of individuals from year to year (e.g.\ married, divorced, married with children), and is correlated from year to year. We expect it to be efficiently captured by a classical HMM due to the local nature of the data, and use it as a control. The Promoter Gene Sequences data set consists of DNA sequences that encode promoters and non-promoters, and therefore has a less obvious local structure. Finally, the SPECT Heart data set encodes binary feature vectors of heart images, with little to no local correlations. 

To estimate the generalization performance of the models, we withheld a quarter of the data for testing the biofam and Promoter Gene Sequences data set, and used the standard SPECT Heart testing data set. Our results are summarized in Fig.~\ref{fig:performance_plot}, where we plot the stochastic estimate of the KL divergence $\tilde{D}_{\text{KL}}$ (normalized by the sequence length) as a function of the local hidden dimension $k$. As we are interested in the optimal performance over all parameters to compare the expressive power of quantum models versus classical models, we plot the minimum achieved loss over ten trials. In particular, for the Promoter Gene Sequences and SPECT Heart data sets, the basis-enhanced $2$-gram model learns the distribution of samples more effectively and also generalizes more effectively than the classical HMM. As expected, both models perform equally well on the biofam data set, since it has very local correlations. These results also demonstrate that for data sets that have no obvious local structure, quantum models tend to perform better, which is consistent with our theoretical analysis in Sec.~\ref{sec:hmm}. Furthermore, we performed a likelihood-ratio test between the two models to measure the statistical significance of the improvement in performance, accounting for any potential overfitting due to the quantum model having more parameters than the classical model. Taking the null hypothesis that the optimal parameters of the basis-enhanced $2$-gram model reduce the model to a classical hidden Markov model, we found using the observed difference in achieved KL divergence that the null hypothesis can be rejected with $5\sigma$ confidence on the Promoter Gene Sequences and SPECT Heart data sets (see Fig.~\ref{fig:p_value_plot}).

Interestingly, the performance separation between quantum and classical models persists even when considering the average performance over many runs. These results are summarized in Appendix~\ref{app:avg_perf}. In combination, these sets of numerical results show that quantum models with nonclassical correlations do have better performance as generative models on real world data. The performance boost is also robust to practical training procedures and realistic performance metric considerations. 

\section{Conclusion and Outlook}\label{sec:outlook}\label{sec:beyond}\label{sec:beyondBayesian}
In this work, we have presented unconditional proof of the separation in expressive power between Bayesian networks and their minimal extension, basis-enhanced Bayesian quantum circuits. We showed that the origin of this separation is associated with quantum nonlocality and contextuality. Focusing on sequential models, we constructed examples via quantum nonlocality of a linear separation in $k$ between $k$-gram models and their basis-enhanced version, and through quantum contextuality, a quasi-polynomial separation in bond dimension for hidden Markov model and its basis-enhanced version. In addition, we numerically tested this separation on standard data sets, showing that this separation holds even on practical data sets.

Although we focused on Bayesian networks, 
our approach can also be applied to more general models. Contextuality provides a general framework since the error model in Eq.~\eqref{eq:error} is independent of the normalization of probability distributions; therefore, our techniques can be applied to graphical models without well-defined transition probabilities along some edges of the graph. For example, Theorem~\ref{thm:hmm} also works for deep Boltzmann machines,
which is a much harder model than Bayesian networks in terms of computational cost. However, there is an intrinsic difficulty in extending Theorem~\ref{thm:hmm} to get a separation with some non-energy-based neural networks (e.g., CNN, RNN~\cite{goodfellow2016deep}). The reason is that one hidden neuron in such kinds of models can take values over real numbers and thus it could potentially carry infinite information, and our counting methods used in the proof of Theorem~\ref{thm:hmm} do not directly apply. The possible extensions and applications of our approach to such models deserve further theoretical investigations.

Our results establish a powerful connection between quantum foundations and machine learning. Since many traditional machine learning models are based on the understanding and intuition from classical physics, they can be naturally characterized by ontological models. Our study shows that quantum correlation can be a resource to enhance the efficiency of these models even if the task is purely classical (e.g.\ by noting the similarity between contextuality in natural languages and quantum contextuality). Our work opens new avenues for using ideas from quantum foundations to develop novel  machine learning models based on MPS, tree-like tensor network, or the Multi-scale Entanglement Renormalization Ansatz (MERA)~\cite{stoudenmire2016supervised,PhysRevX.8.031012,1803.09111,glasser2020probabilistic,1803.10908,stoudenmire2017learning}. In addition, we expect the concepts of quantum correlations can be used to provide theoretical foundations for other quantum-inspired classical models and quantum machine learning models.

Finally, our work provides new insights into designing practical quantum machine learning algorithms
%, and makes major advancements in looking for structured and practical quantum circuits 
that exhibit quantum advantage in tackling machine learning tasks; this can be achieved by starting from a successful classical machine learning model and enhancing it with quantum correlations. Although the examples we used here can be efficiently simulated classically, some of them require the use of quantum machines during the training stage. Furthermore, the example models we used in this work are subclasses of more general sequential quantum generative models, involving quantum circuits with sequential (adaptive or non-adaptive) measurements, which have to be implemented on quantum hardware. It can be expected that the ideas of contextuality presented here can be extended to these cases to achieve an even stronger quantum advantage. 

\begin{acknowledgments}
We thank Stephen Bartlett, Dongling Deng, Zhengfeng Ji, Jordi Tura, and Seth Lloyd for helpful discussion. X.G.\ is supported by the Postdoctoral Fellowship in Quantum Science of the Harvard-MPQ Center for Quantum Optics, the Templeton Religion Trust grant TRT 0159, and by the Army Research Office under Grant W911NF1910302 and MURI Grant W911NF-20-1-0082.
E.R.A.\ is supported by the National Science Foundation Graduate Research Fellowship Program under Grant No.\ 4000063445, and a Lester Wolfe Fellowship and the Henry W.\ Kendall Fellowship Fund from MIT.
S.-T.W.\ is partially supported by Air Force STTR under grant No.\ FA8750-20-P-1708. J.I.C.\ acknowledges funding QUENOCOBA, ERC-2016-ADG (grant no.\ 742102) and the D-A-CH Lead-Agency Agreement through project No. 414325145(BEYOND C). M.D.L.\ acknowledges funding NSF CUA - PHY-1125846,
NSF - PHY-2012023,
ARO MURI - W911NF2010082,
DARPA ONISQ - W911NF2010021.
\end{acknowledgments}

\appendix

\section{Relations Among Various Machine Learning Models}\label{app:relation_models}
Deep belief nets have the form
\begin{equation}
\begin{aligned}
&p\left(\bm{v}, \bm{h_1}, \bm{h_2}, \ldots,\bm{h_d}\right)=\\
&p\left(\bm{v}\mid\bm{h_1}\right)
\ldots p\left(\bm{h_{d-2}}\mid\bm{h_{d-1}}\right)p\left(\bm{h_{d-1}},\bm{h_{d}}\right);
\end{aligned}
\end{equation}
this is exactly the form of a $2$-gram model if the hidden variables are also observed.

A $k$-gram model can be simulated by an HMM with $L^{k-1}$ hidden variables per site, where $L$ is the vocabulary length of the $k$-gram model; this can be done straightforwardly by combining sets of $k-1$ sites in the $k$-gram model into one site in the HMM. The $L^{k-1}$ possible values of these sites in the $k$-gram model map to each of the $L^{k-1}$ hidden variables in the HMM.

\section{Details of Basis-Enhanced Bayesian Quantum Circuits}

\subsection{Mapping Between Bayesian Networks and Quantum Circuits}\label{app:equivalence}

\begin{figure*}[tp]
\includegraphics[width=\textwidth]{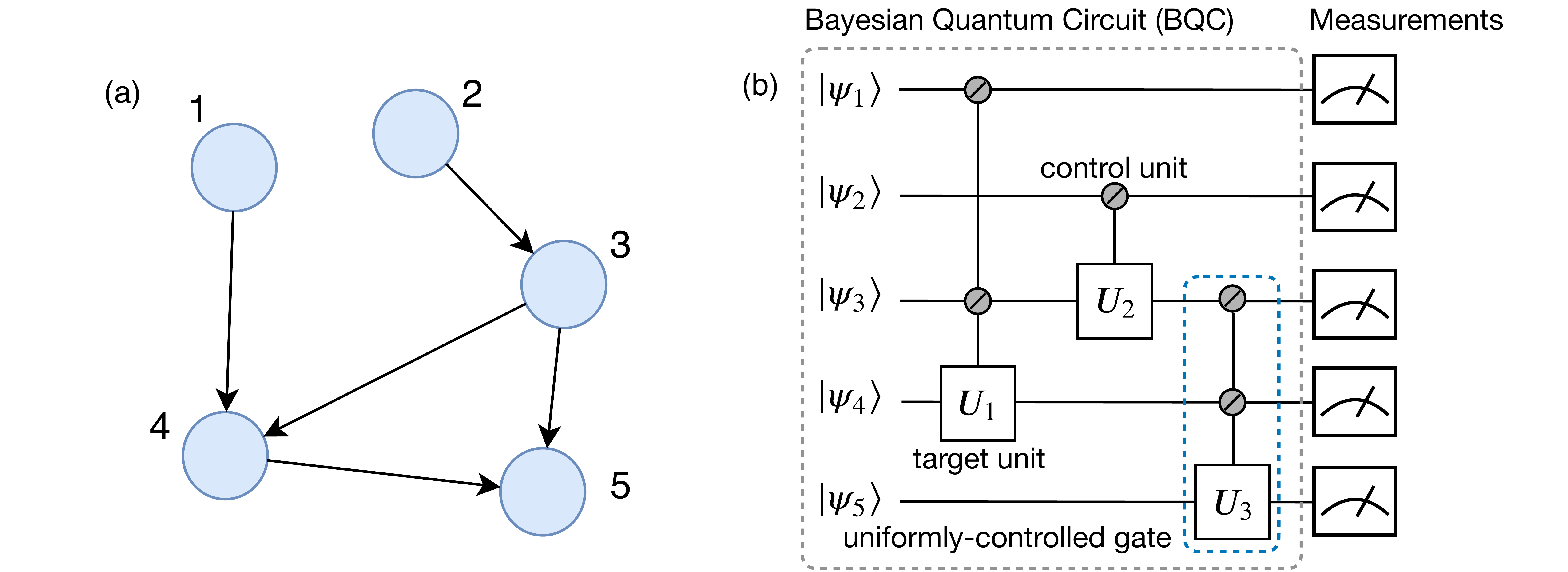}
\caption{A Bayesian network on a generic graph and its associated Bayesian quantum circuit. (a) An example of a Bayesian network on a generic graph. Each node $x_i$ corresponds to a transition probability $p\left(x_i\mid\text{parents of }x_i\right)$, e.g. the node $x_4$ corresponds to $p\left(x_4\mid x_1,x_3\right)$. If there is no parent for the node $x_i$, it simply corresponds to a marginal probability $p\left(x_i\right)$, e.g. the node $x_2$ corresponds to $p\left(x_2\right)$. The joint probability distribution defined by the Bayesian network is the product of the transition or marginal probability over all of the nodes. For this example, the probability distribution is $p\left(x_1,x_2,x_3,x_4,x_5\right)=p\left(x_1\right)p\left(x_2\right)p\left(x_3\mid x_2\right)p\left(x_4\mid x_1,x_3\right)p\left(x_5\mid x_3,x_4\right)$. (b) The associated Bayesian quantum circuit. Each node in the Bayesian network corresponds to a qubit and each uniformly controlled gate corresponds to a transition probability. For example, the uniformly control-$U_1$ gate corresponds to $p\left(x_4\mid x_1,x_3\right)$. The marginal probabilities $p\left(x_1\right),p\left(x_2\right)$ have been absorbed into $\ket{\psi_1},\ket{\psi_2}$. Because all Bayesian networks are associated with directed acyclic graphs, an ancestor of a node cannot also be a child of the same node. This means, in Bayesian quantum circuits, there cannot be a target unit after a control unit. Furthermore, the target unit can only involve one qubit in a uniformly controlled gate. Measuring the output qubits of Bayesian quantum circuits in the computational basis will produce the same probability distribution as the corresponding Bayesian network.
}
\label{fig:baynet2}
\end{figure*}

In the following, we give the explicit construction of the mapping between Bayesian networks and quantum circuits (see Fig.~\ref{fig:baynet2} as an illustration):
\begin{itemize}
\item Bayesian networks $\Longrightarrow$ BQCs. Each node corresponds to a qubit. According to the direction of edges in the graph, assign an order for these qubits. Then do the following steps in order. If the node has no parent, prepare the corresponding qubit as $\ket{\psi_i}$ such that $\left\lvert\innerp{x_i}{\psi_i}\right\rvert^2=p\left(x_i\right)$. Otherwise, prepare the corresponding qubit as $\ket{\psi_i}$ and apply $U(\text{parents of }x_i)$ on the corresponding qubits such that
\begin{equation}\label{eq:transition}
\left\lvert\bra{x_i} U\left(\text{parents of }x_i\right)\ket{\psi_i}\right\rvert^2=p\left(x_i\mid\text{parents of }x_i\right).
\end{equation}
Notice that there are no other operations between $U\left(\text{parents of }x_i\right)$ and $\ket{\psi_i}$ since (i) if there were a target unit $V$ in another uniformly controlled gate, $U$ and $V$ could be merged into a single uniformly controlled gate; (ii) the order guarantees there is no control unit before a target unit.

\item BQCs $\Longrightarrow$ Bayesian networks. First, we assign a directed acyclic graph to the BQC. Each qubit corresponds to a node and we draw an arrow from node $x_i$ to node $x_j$ if and only if there exists a uniformly controlled gate with control unit on qubit $i$ and target unit on qubit $j$. Then we assign each node transition probabilities in the following way: if there is no target unit on qubit $i$, assign $p\left(x_i\right)=\left\lvert\innerp{x_i}{\psi_i}\right\rvert^2$ for the corresponding node; otherwise, assign a transition probability for this node according to Eq.~\eqref{eq:transition}.
\end{itemize}

\subsection{Efficient Implementation Using Multi-Qubit Collective Gates}\label{app:MRG}

The implementation of uniformly controlled gates is not efficient in general~\cite{bergholm2005quantum}. This is true even if we have the ability to implement collective gates which are native, for instance, to Rydberg-based quantum platforms (e.g. implementing quantum fan-out gates with $k$ control units \cite{isenhower2011multibit} as shown in Fig.~\ref{fig:transition}(a)). Even though these gates are very powerful~\cite{hoyer2005quantum}, it is unclear how to implement general uniformly controlled gates more efficiently (in terms of scaling with $k$). However, in almost all machine learning models, the transition probabilities have specific forms when $k$ is large. In particular, transition probabilities usually take the form: 
\begin{equation}\label{eq:linear}
p\left(x_{k+1}\mid x_1,\dots,x_k\right)= f\left( x_{k+1} \left( \sum_{i=1}^k w_i x_i \right)\right);
\end{equation}
that is, the dependence on the parent variables is linear in a non-linear function $f$ in general. For example, in deep belief nets, $f\left(y\right)=e^{-\beta y}/\left(1+e^{-\beta y}\right)$. Here we give a construction showing how to implement such uniformly controlled gates approximately such that the number of elementary collective gates does not depend on $k$.

\begin{figure*}
\includegraphics[width=\textwidth]{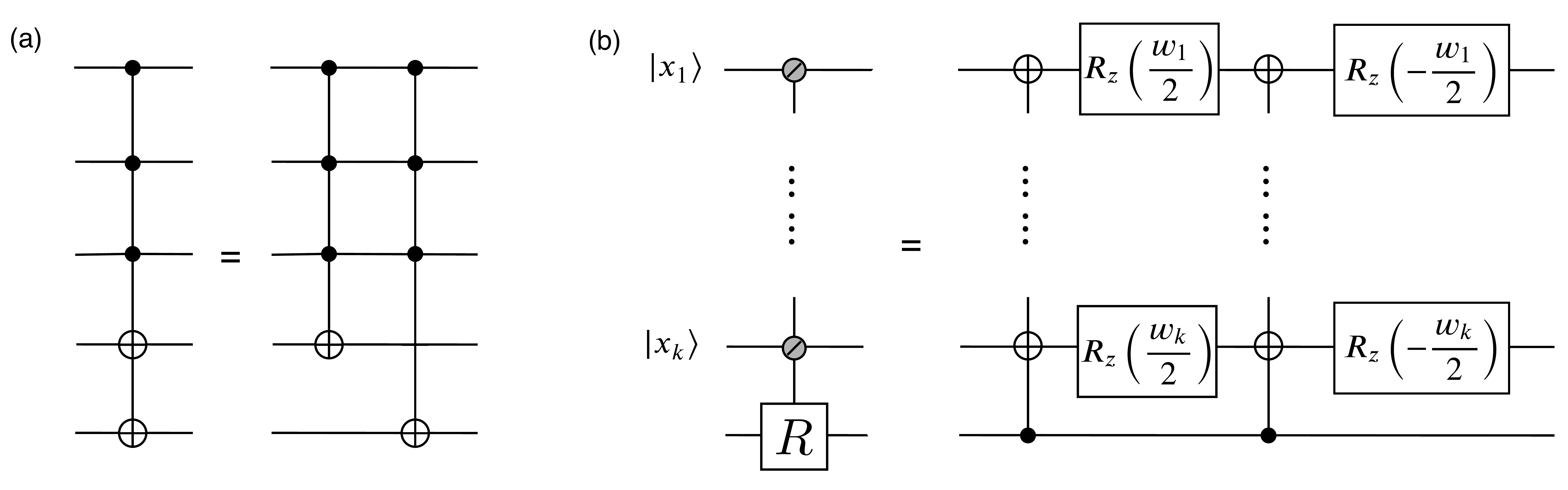}
\caption{Basic elements in Bayesian quantum circuits implemented by multi-qubit collective gates. (a) ``Basic'' collective gate. (b) Implementation of control-$R=R_z\left(\sum\limits_iw_ix_i\right)$ gate. This gate can be implemented directly on a Rydberg-atom based platform.
}
\label{fig:transition}
\end{figure*}
In the following, we will show how to implement the transition probability shown in Eq.~\eqref{eq:linear} in BQCs with a circuit depth independent of $k$, using collective gates. We define $\theta=\sum w_ix_i$, such that according to Eq.~\eqref{eq:linear}, $p\left(0\mid x_1,\ldots,x_k\right)=f\left(0\right)$ and $p\left(1\mid x_1,\ldots,x_k\right)=f\left(\theta\right)$. Thus we have a normalization condition $f\left(0\right)+f\left(\theta\right)=1$. We introduce the notation $\left\langle\cdot\right\rangle_{m}$ as a binary representation of $\cdot$ up to the $m$-th digit. We also introduce $\tilde{\theta}$ as an approximation of $\theta$, with binary representation $\left\langle\theta\right\rangle_{d_1}$. We then use the following procedure to implement the transition:
\begin{widetext}
%$\phi$ is an approximation of $\arccos{\sqrt{f(\tilde\theta)}}$ of which binary representation is exactly $<\arccos{\sqrt{f(\tilde\theta)}}>_{d_2}$
\begin{eqnarray}
\ket{0}^{\otimes d_1}\ket{0}^{\otimes d_2}\ket0 
&\longrightarrow &
\ket{\left\langle\theta\right\rangle_{d_1}} \ket{0}^{\otimes d_2}\ket0 
\text{\quad(phase estimation algorithm, }O(d_1^2)\text{ gates)}
\notag\\
&\longrightarrow&
\ket{\left\langle\theta\right\rangle_{d_1}} \ket{\left\langle\arcsin{\sqrt{f(\tilde\theta)}}\right\rangle_{d_2}}\ket0
\text{\quad(classical computing, usually }O\left(\operatorname{poly}\left(d_2^2\right)\right)\text{ gates)}
\notag\\
&\longrightarrow&
\ket{\left\langle\theta\right\rangle_{d_1}} \ket{\left\langle\arcsin{\sqrt{f(\tilde\theta)}}\right\rangle_{d_2}}
\left(\sqrt{f(0)}\ket0+\sqrt{f(\theta)}\ket1\right)
\text{\quad(controlled rotation along }x\text{ axis)}
\notag\\
&\longrightarrow&
\ket{0}^{\otimes d_1}\ket{0}^{\otimes d_2}\left(\sqrt{f(0)}\ket0+\sqrt{f(\theta)}\ket1\right)
\text{\quad(uncomputing)}.
\end{eqnarray}
\end{widetext}
The precision of the transition is determined by $d_1$ and $d_2$. $d_1$ determines the precision of the input of function $f$ and $d_2$ determines the effect of truncation for the function $\arcsin\sqrt{ f(\cdot)}$. The total error is bounded by
\begin{equation}
\epsilon=\max_\theta \left\lvert \left( \arcsin\sqrt{ f(\theta)} \right)^\prime \right\rvert 2^{-d_1}+2^{-d_2}.
\end{equation}
Therefore, if the derivative is bounded by a constant (for example, in the case of $f(\theta)=e^{-\beta\theta}/(1+e^{-\beta\theta})$, the derivative is bounded by $\beta$), $d_1$ and $d_2$ can be taken to be $O\left(\log\left(1/\epsilon\right)\right)$ such that the depth is bounded by $\operatorname{poly}\left(\log\left(\frac{1}{\epsilon}\right)\right)$,
which is independent of $k$ and only depends on the precision $\epsilon$.

\subsection{Exponential Separation of Expressive Power Between BBQCs and Bayesian Networks Based on Computational Complexity Theory}\label{app:general_graph}
The proof of the exponential expressive power of BBQCs is a slight modification of the proof for Quantum Generative Models (QGMs) detailed in Ref.~\cite{gao2018quantum}.
\begin{theorem}[Ref.~\cite{gao2018quantum}]
There exists a BBQC with $n$ qubits such that, if any Bayesian networks with a polynomial number of parameters in $n$ could approximate it under multiplicative error, the polynomial hierarchy in computational complexity theory would collapse.
\end{theorem}
For completeness, we give a brief review of the proof. First, we give a brief introduction of related concepts. Second, we introduce a specific BBQC which is used to separate the expressive power between the classical and quantum models. Third, we give a sketch of the proof. See Ref.~\cite{gao2018quantum} for more details.

\subsubsection{Related Computational Complexity Classes}

The polynomial hierarchy is a hierarchy of complexity classes that generalizes \textsf{P} and \textsf{NP}, and are denoted as $\mathsf{\Sigma}^p_0,\mathsf{\Sigma}^p_1,\mathsf{\Sigma}^p_2,\ldots$. Here, $\mathsf{\Sigma}^p_0=\mathsf{P}$, $\mathsf{\Sigma}^p_1=\mathsf{NP}$ and $\mathsf{\Sigma}^p_{i+1}=\mathsf{NP}^{\mathsf{\Sigma}^p_i}$ where $\mathsf{NP}^{\mathsf{\Sigma}^p_i}$ is called \textsf{NP} relative to $\mathsf{\Sigma}^p_i$. \textsf{NP} denotes problems which can be verified in polynomial time by a Turing machine and $\mathsf{NP}^{\mathsf{\Sigma}^p_i}$ denotes problems which can be verified in polynomial time by a Turing machine that is equipped with an oracle which can solve any $\mathsf{\Sigma}^p_i$ problems in one step. A detailed discussion can be found in~\cite{arora2009computational} or in the recent review article
on quantum supremacy~\cite{Harrow2017}. It is widely believed that the polynomial hierarchy does not collapse which means $\mathsf{\Sigma}^p_i\neq \mathsf{\Sigma}^p_{i+1}$ (which implies $\mathsf{\Sigma}^p_i\neq \mathsf{\Sigma}^p_{i+j}$ for any constant $j>0$).

\subsubsection{The Basis-Enhanced Bayesian Network Used in the Proof}\label{app:qs}
Here we give a construction of a basis-enhanced Bayesian network such that approximately computing the probability of a specific configuration up to multiplicative error is \textsf{\#P}-hard.

\begin{figure}
\centering
\includegraphics[width=0.8\linewidth]{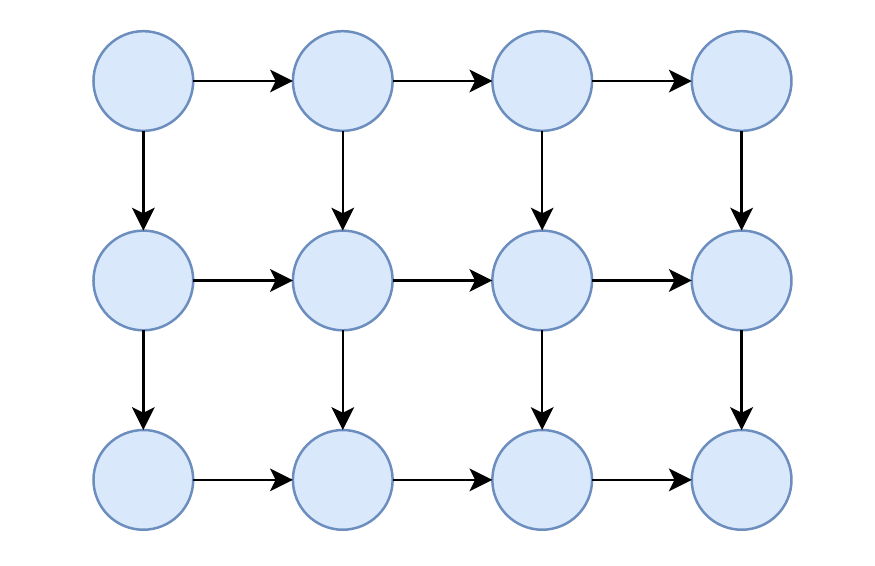}\caption{A cluster state in a Bayesian quantum circuit. A cluster is generated from initial state $\ket+^{\otimes n}$ by applying control-Z gates between each pair of neighbors on a square lattice. We may assign each edge an arrow in the way shown in this figure in order to get a directed acyclic graph. We can see that this circuit is a Bayesian quantum circuit by checking Definition~\ref{def:BQN} and noticing that, in a control-Z gate, there is no need to distinguish between control and target units. %\textsf{\#P}-hardness for quantum extension of Bayesian networks: the corresponding Bayesian quantum circuit and  in Ref.\cite{gao2016quantum}. Here is a conventional representation for graph state where each circle corresponds to an identity tensor and each edge corresponds to a Hadamard tensor as shown in Fig.\ref{fig:QGM}.
%To construct this state, we start from a brickwork of white circles \cite{broadbent2009universal} (the top side), with each white circle
%representing seven blue circles. Each blue circle represents a qubit. Then we apply $HZ(\theta)$ which is clearly an invertible matrix on each qubit with the angle $\theta$ shown at the bottom. The idea of the proof is through measurement-based quantum computing \cite{raussendorf2001one} and the equivalence between quantum computing with post-selection and \textsf{\#P} \cite{aaronson2005quantum}, and is widely used to prove quantum supremacy of sampling from dynamics of Ising interaction \cite{bremner2011classical}.
}%
\label{fig:cluster}%
\end{figure}

This BBQC begins as a cluster state on a square lattice. The corresponding Bayesian network is drawn as the graph shown in Fig.~\ref{fig:cluster}. Then, we use the measurement basis shown in Ref.~\cite{gao2016quantum}. One of the important properties of this construction is its ``single-instance-hardness,'' which means there is only one measurement basis for any fixed size; i.e.\ the probability distribution $q\left(\bm{x}\right)$ only depends on the size of the lattice. We demand this property because in the proof of exponential expressive power, we will associate a probability distribution with a problem with output that are nonnegative numbers: $\bm{x}$ specifies an instance of the problem and the task is to compute the probability given a specific $\bm{x}$, i.e.\ $q\left(\bm{x}\right)$, to multiplicative error. Thus the complexity of a probability distribution is defined as the complexity of the associated problem.

The proof of exponential expressive power works for any efficiently computable classical model. Thus it also works for any neural networks.

\subsubsection{Sketch of the Proof}

The key to separating the complexity of the classical and quantum models is formalizing a sign problem caused by quantum interference: approximately computing (up to multiplicative error) a summation of many nonnegative numbers is easier than the summation of many complex or real numbers. This can be done via Stockmeyer's theorem~\cite{stockmeyer1985approximation} (see~\cite{gao2018quantum} for an introduction oriented to the proof here); the former is inside $\mathsf{\Sigma}^p_2$ and the latter is \textsf{\#P}-hard. The same reasoning has been used to separate QGMs and general probabilistic graphical model~\cite{gao2018quantum}. Here we only give a sketch of the proof.

Assume there exists a Bayesian network which generates the joint probability $p\left(\bm{x},\bm{y},\bm{z}\right)$ such that the conditional probability $\sum_{\bm{y}}p\left(\bm{x},\bm{y}\mid\bm{z}\right)$ approximates $q\left(\bm{x}\right)$ to multiplicative error. We can use Stockmeyer's theorem to prove that, based on this assumption and supposing that the parameters of the network are given, approximating $q\left(\bm{x}\right)$ to multiplicative error is in $\mathsf{\Sigma}^p_2$.
We should keep in mind that the probability defines a problem with $\bm{x}$ specifying an instance of the problem.

%According to Stochmeyer's theorem, there exists an $\mathsf{FBPP}^\mathsf{NP}$ algorithm approximating $P_1=\sum_{\mathbf x,\mathbf y}p(\mathbf x,\mathbf y,\mathbf z)$ by $\widetilde P_1$ such that $|P_1-\widetilde P_1|\le\gamma_1P_1$ and $P_2=\sum_{\mathbf y}p(\mathbf x,\mathbf y,\mathbf z)$ by $\widetilde P_2$ such that $|P_2-\widetilde P_2|\le\gamma_2P_2$. Define $\widetilde p(\mathbf x)=\widetilde P_2/\widetilde P_1$ and we have $p(\mathbf x)=P_2/P_1$.

%the last step is because we choose $\gamma_1$ and $\gamma_2$ as sufficiently small as $1/\text{poly}(n)$. 
%Under the assumption that the representation is efficient, such $f(\mathbf x,\mathbf y)$ can be represented by a polynomial size circuit, where the circuit corresponds to the description of the factor graph. So $\widetilde p(\nonumber x)$ can be computed in $\mathsf{FBPP}^\mathsf{NP}/\text{poly}$.

However, though we can show that it is possible to approximate $q$ to multiplicative error, the proof is not constructive. More concretely, ``$/\operatorname{poly}$'' denotes that, for any fixed input size (the length of $\bm{x}$), there exists a polynomial sized classical circuit that computes all the instances of the problem, but the circuit may not be efficiently constructed~\cite{karp1982turing,arora2009computational}. In Appendix~\ref{app:qs}, we construct a BBQC such that computing $q\left(\bm{x}\right)$ to multiplicative error is $\mathsf{\#P}$-hard. Thus, assuming the efficient representation of the BBQC via classical Bayesian networks, we (roughly) obtain
\begin{equation}
\mathsf{\#P}\subseteq  \mathsf{\Sigma}^p_2/\operatorname{poly}.
\end{equation}
This implies that the polynomial hierarchy would collapse to the third level, as more formally shown in Ref.~\cite{gao2018quantum} (which follows from a modification of the reasoning of the proof of Theorem 3 in Ref.~\cite{aaronson2017implausibility}).

\subsection{Algorithms for Inference and Learning}\label{app:quantum_algorithm}

There are mainly two computational problem associated with a generative model.
One is inference, i.e.\ how to extract useful information from the representation of the generative model. Making inference on a generative model usually means computing marginal probabilities, conditional probabilities, or performing maximum likelihood estimation. With this, we can make predictions for new data after getting an approximately correct representation of a data distribution $p_{\mathcal D}$. Later on, we will give examples to show the applications of computing conditional probabilities.

The other is training (or learning), i.e.\ how to determine parameters of the generative model from training data in order to approximate $p_{\mathcal D}$. Training usually means minimizing the
KL divergence:
\begin{equation}
D(p_{\mathcal D}||p_{\bm\theta})=\sum_{\mathbf v}p_{\mathcal D}(\mathbf v)\log\left(\frac{p_{\mathcal D}(\mathbf v)}{p_{\bm\theta}(\mathbf v)}\right)
\end{equation}
between $p_{\mathcal D}$ and $p_{\bm\theta}$, the distribution of the generative model, with the whole parameter set denoted by $\bm{\theta}=\left(\theta_1,\cdots,\theta_{\operatorname{poly}\left(n\right)}\right)$. The $\bm\theta$-dependent part of $D\left(p_{\mathcal D}\mid\mid p_{\bm\theta}\right)$ can be expressed as 
\begin{equation}\label{eq:stoch_kl_div}
D\left(\bm\theta \right)\equiv-\frac{1}{N}\sum_{\bm{v}\in \text{training data set}}\log p_{\bm\theta }\left(\bm{v}\right),
\end{equation} 
where $N$ denotes the total number of data ($1/N$ approximates $p_{\mathcal D}\left(\bm{v}\right)$) and the summation is over all the training data or a batch of data for stochastic optimization. As the number of parameters is bounded by $\text{poly}\left(n\right)$,
the required data size $N$ is typically bounded by $\text{poly}\left(n\right)$~\cite{shalev2014understanding}. We can also understand optimizing $D\left(\bm \theta\right)$ as maximum likelihood estimation since $D\left(\bm\theta \right)$ is proportional to the log-likelihood $\log\left(\prod_{\mathbf v} q_{\bm\theta}\left(\bm{v}\right)\right)$. It is worth mentioning that, in addition to $D\left(\bm{\theta}\right)$ with $\bm{v}=\left(\bm{x},y\right)$, it is also usual to adopt the following loss function for supervised learning~\cite{larochelle2008classification}:
\begin{equation}
L\left(\bm{\theta}\right)\equiv-\sum_{\bm{v}\in \text{training data set}}\log q_{\bm\theta }\left(y\mid\bm{x}\right),
\end{equation}
since it is the log-likelihood $\log\left(\prod_{\bm v} q_{\bm\theta }(y\mid\bm{x})\right)$. Typically, we
minimize these loss functions via a so-called optimizer, usually the gradient descent method~\cite{goodfellow2016deep} with a proper learning rate (the step length for updating parameters) or its variations like adding a stochastic term, adjusting the learning rate adaptively, utilizing training data in batches, and so on.

\subsubsection{Heuristic Quantum Algorithms for Inference and Learning}
Even for classical Bayesian networks, the training and inference problems are computationally hard for quantum computers~\cite{chickering1996learning,dagum1993approximating}. However, there are a number of proposed heuristic and approximate algorithms that works well in practice in tackling these computation problems for classical algorithms~\cite{koller2009probabilistic}.

BBQCs have a similar problem in that exact training and inference are computationally difficult, and it is natural to propose heuristic quantum algorithms. Since Bayesian networks are a special case of probabilistic graphical models~\cite{koller2009probabilistic}, the quantum algorithm for the learning and inference problems in extensions of probabilistic graphical models~\cite{gao2018quantum} also works here. The idea is to convert the learning and inference problems to preparing ground states of a local Hamiltonian. The runtime of the quantum algorithm is proportional to the inverse of the energy gap, although we cannot guarantee that the energy gap scales as $1/\operatorname{poly}\left(n\right)$. Other heuristic quantum algorithms more specific to Bayesian networks may also exist as in the classical case.

\section{Relations Among Various Error Models}\label{app:errors}

For completeness, we also define two other error models that we refer to. One is the multiplicative error:
\begin{equation}
\left\lvert p\left(\bm{x}\right)-q\left(\bm{x}\right)\right\rvert\le\gamma q\left(\bm{x}\right), \forall \bm{x}
\end{equation}
with $\gamma$ being a constant smaller than $1/2$. $p$ approximating $q$ under this error implies that $D_{\text{KL}}\left(p\mid\mid q\right)$ is bounded by $\gamma$.
Thus $p$ approximating $q$ under this error model is a stronger requirement than a small KL divergence of $p$ from $q$. This error model is used for our complexity theory based proof of quantum advantage on general graphs in Appendix~\ref{app:general_graph}.

For translation problems, the generative models usually only define conditional probabilities $p\left(\bm{y}\mid\bm{x}\right)$ for the classical model (e.g.\ the second model in Fig.~\ref{fig:baynet}(b)) and $q\left(\bm{y}\mid\bm{x}\right)$ for the quantum extension. The prior probability for $\bm{x}$ is unspecified and we denote it as $p\left(\bm{x}\right)$ and $q\left(\bm{x}\right)=p\left(\bm{x}\right)$. Then, the KL divergence is
\begin{eqnarray}
D_\text{KL}(p\mid\mid q)&=&\sum_{\bm{x},\bm{y}}p\left(\bm{x},\bm{y}\right)\log\frac{p\left(\bm{x},\bm{y}\right)}{q\left(\bm{x},\bm{y}\right)}\\\nonumber
&=&\sum_{\bm{x}}p\left(\bm{x}\right)\left(\sum_{\bm{y}}p\left(\bm{y}\mid\bm{x}\right)\log\frac{p\left(\bm{y}\mid\bm{x}\right)}{q\left(\bm{y}\mid\bm{x}\right)}\right).
\end{eqnarray}
In order to avoid any assumptions on $p\left(\bm{x}\right)$, we require the quantity inside the brackets be bounded for any $\bm{x}$. This implies that
\begin{equation}\label{eq:error2}
q\left(\bm{y}\mid\bm{x}\right)=0\iff p\left(\bm{y}\mid\bm{x}\right)=0, \forall \bm{x}, \bm{y}.
\end{equation}

Now, let us show that the multiplicative error is bounded implies that $D_\text{KL}\left(p\mid\mid q\right)$ and $D_\text{KL}\left(q\mid\mid p\right)$ are bounded, which in turn implies the error model in Eq.~\eqref{eq:error}. We see that:

\begin{eqnarray}
\nonumber\left\lvert \sum_{\bm{x}}p\left(\bm{x}\right)\log \frac{p\left(\bm{x}\right)}{q\left(\bm{x}\right)}\right\rvert 
&\le&\sum_{\bm{x}} \left\lvert p\left(\bm{x}\right)\log\left( 1+\frac{p\left(\bm{x}\right)}{q\left(\bm{x}\right)}-1\right)\right\rvert
\\
 &\leq
&
\sum_{\bm{x}}
\left\lvert p\left(\bm{x}\right)\left(\frac{p\left(\bm{x}\right)}{q\left(\bm{x}\right)}-1\right)\right\rvert
\notag \\
\nonumber&\leq &\sum_{\bm{x}}p\left(\bm{x}\right)\gamma \\
&=&\gamma.
\end{eqnarray}
The second inequality comes from
\begin{equation}
\left\lvert p\left(\bm{x}\right)-q\left(\bm{x}\right)\right\rvert\leq\gamma q\left(\bm{x}\right)
\end{equation}
with small $\gamma$. A similar proof holds for $D_\text{KL}\left(q\mid\mid p\right)$ by exchanging $p$ and $q$.

According to the definition of $D_\text{KL}\left(p\mid\mid q\right)$, $q\left(\bm{x}\right)=0$ implies that $p\left(\bm{x}\right)=0$ in order to make $D_\text{KL}\left(p\mid\mid q\right)$ bounded. According to the definition of $D_\text{KL}\left(q\mid\mid p\right)$, $p\left(\bm{x}\right)=0$ implies that $q\left(\bm{x}\right)=0$ in order to make $D_\text{KL}\left(q\mid\mid p\right)$ bounded. Thus both $D_\text{KL}\left(p\mid\mid q\right)$ and $D_\text{KL}\left(q\mid\mid p\right)$ are bounded implies the error model in Eq.~\eqref{eq:error}.

\section{Lemma Proofs for the Mermin--Peres Magic Square}\label{app:stab}

\begin{table*}[t!]
\caption{Mermin--Peres magic square}
\begin{center}
\begin{tabular}{|c|c|c|c|}
\hline
The computational states & $\left(-1\right)^{f}I_1Z_2Z^{\bm a_3+\bm c_3}$ &  $\left(-1\right)^{g}Z_1I_2Z^{\bm b_3+\bm d_3}$ &  $\left(-1\right)^{f+g}Z_1Z_2Z^{\bm a_3+\bm b_3+\bm c_3+\bm d_3}$\\
\hline
The first graph state & $X_1I_2Z^{\bm a_3}$ & $I_1X_2Z^{\bm b_3}$ & $X_1X_2Z^{\bm a_3+\bm b_3}$\\
\hline
The second graph state & $X_1Z_2Z^{\bm c_3}$ & $Z_1X_2Z^{\bm d_3}$ & $-X_1X_2 Z_1Z_2Z^{\bm c_3+\bm d_3}$\\
\hline
\end{tabular}
\end{center}
\label{tab}
\end{table*}%

First, we prove the following lemma:

\begin{lem}\label{lem:square}
For any subset of stabilizer states $s$ such that $\left|s\right|>2^{n^2/4+7n/2}$, there exist three states such that some of their stabilizers forming a Mermin-Peres magic square as shown in Table~\ref{tab}.
\end{lem}

\begin{proof}
We write $\left|s\right| =a\left(n\right)b\left(n\right)c\left(n\right)+1$ where $a\left(n\right)=2^{n^2/4+3n/2}$, $b\left(n\right)=4^n$ and $c\left(n\right)=2$; their meaning will be clear later. Given $\left|s\right|$ stabilizer states, we can always transform these states to another set of stabilizer states by a Clifford circuit such that one of the state will become $\ket{0}^{\otimes n}$ and the other $a\left(n\right)b\left(n\right)c\left(n\right)$ states have the following form (see Ref.~\cite{nest2008classical}):
\begin{equation}
\ket{\psi}\propto \sum_{x\in A}\left(-1\right)^{q\left(x\right)}i^{l\left(x\right)}\ket{x},
\end{equation}
where $A$ is an affine subspace of $\mathbb{Z}_2^n$ and $q\left(x\right)$ and $l\left(x\right)$ are quadratic and linear functions on $\mathbb Z_2$ and $\mathbb Z_4$, respectively. The state is determined by $A, q, l$.  

We denote $a\left(n\right)$ as the number of different $A$: an affine subspace is composed of a linear subspace and a displacement, and thus there are at most 
\begin{equation}
a\left(n\right)\leq\sum_k 2^{\left(n-k+1\right)k}\times 2^n=2^{n^2/4+3n/2}
\end{equation}
possible $A$, where the first term involving summation over $k$ is the number of linear subspaces (where $k$ the dimension of the subspace---see Theorem~2.14 of~\cite{z2_vec_space}) and the second term $2^n$ is the number of possible displacements.

According to the pigeonhole principle, we can prove that we now have the $\ket{0}^{\otimes n}$ state and at least $b\left(n\right)c\left(n\right)$ states belonging to the same affine subspace $A$. Those $b\left(n\right)c\left(n\right)$ states only differ by $q$ and $l$.
Using $C_\text{CNOT}$ and $C_\text{X}$ which are circuits only composed of CNOT and Pauli X gates respectively, these states can be transformed to be of the form \begin{equation}\label{eq:graph_state}
\sum_{u\in\left\{0,1\right\}^k}\left(-1\right)^{\bar q\left(u\right)}i^{\bar l\left(u\right)}\ket{u}\ket{0}^{\otimes \left(n-k\right)},
\end{equation}
which are graph states over the first $k$ qubits. These circuits simultaneously transform $\ket{0}^{\otimes n}$ to a state of the form $\ket{z_1z_2\ldots z_n}$.
% Now we have the $\ket{0}^{\otimes n}$ state and at least $b(n)c(n)$ states belonging to the same affine subspace $A$. By applying $C_\text{CNOT}^\dag  C_\text{X}^\dag$, we have a $\ket{z_1z_2\cdots z_n}$ state and at least $b(n)c(n)$ states having the form in Eq.(\ref{eq:graph_state}) on the first $k$ qubits.

Denote $b\left(n\right)$ as the number of $\bar l$ which is no greater than $4^n$ (where the worst case is $k=n$). Then we have the $\ket{z_1z_2\ldots z_n}$ state and at least $c\left(n\right)$ graph states after applying $S$ or $Z$ gates to eliminate $\bar l$ . As long as $c\left(n\right)\ge 2$, we can always find two graph states such that there exists a pair of vertices where there is no edge for the first graph and there is an edge for the second graph. Without loss of generality, we may assume this pair is comprised of qubits $1$ and $2$.
% Then we show that we can always find three states up to a Clifford transformation: the first one is a computational state and the other two are different graph states.

The first graph state (without an edge between qubits $1$ and $2$) has stabilizer generators $X_1I_2Z^{\bm a_3}$ and $I_1X_2Z^{\bm b_3}$ and the second has generators $X_1Z_2Z^{\bm c_3}$ and $Z_1X_2Z^{\bm d_3}$ where $\bm a_3,\bm b_3,\bm c_3,\bm d_3$ are $n-2$ dimensional vectors on $\mathbb Z_2$. The computational state has generators $\left\langle (-1)^{e_i}Z_i\right\rangle$ so it could have any Pauli Z type stabilizers up to $\pm$ signs. Then, we have the Mermin--Peres magic square given by Table~\ref{tab}.
In this table, operators in each row commute with each other since they are chosen from stabilizers of the same quantum states. It is also easy to check that operators in each column commute with each other. The first two Paulis in each observable form the ``traditional'' Mermin square, and the Pauli $Z$s after the first two qubits do not change the commutation relations between observables. Thus, this table forms a Mermin--Peres magic square and thus exhibits contextuality.
\end{proof}

Second, we prove the following lemma:

\begin{lemma}
If three stabilizer states and a subset of their stabilizers form a Mermin--Peres magic square as shown in Table~\ref{tab:square}, the intersection of their support in an ontological theory should be empty in order to be consistent with quantum mechanics.\end{lemma}
\begin{proof}
The proof basically follows the discussion in the main text with the example in  Eq.~\eqref{eq:eg}, but written in a more general way. Assume there is a common $\lambda$ in the intersection among the supports of the three states $\ket{\psi_1}$, $\ket{\psi_2}$, and $\ket{\psi_3}$.

It is simple to show the following two equations:
\begin{eqnarray}
 \notag\frac{1- ABab}{2}\ket{\psi_3}&=&0, \text{ or equivalently }\\ \frac{1+ABab}{2}\ket{\psi_3}&=&\ket{\psi_3},
\end{eqnarray}
  and 
\begin{eqnarray}
\notag \frac{1+ ABab}{2}\ket{\psi_1} &\perp &
\frac{1+ ABab}{2}\ket{\psi_2},\text{ or equivalently }\\\notag
\bra{\psi_1} \frac{1+ ABab}{2} \ket{\psi_2}&=&
\bra{\psi_1}\frac{1+ Aa}{2} \frac{1+ ABab}{2}\frac{1+ Bb}{2} \ket{\psi_2}\\
&=&0.
\end{eqnarray}
The first set of equalities forces the measurement result of $\lambda$ to deterministically be $+1$. The second set of equalities shows that the resulting states of $\ket{\psi_1}$ and $\ket{\psi_2}$ after measuring $ABab$ and getting $+1$ are orthogonal. Thus, there is a contradiction.\end{proof}

\section{Robust Separation of $k$-gram Model Under $l_1$-Distance}\label{app:nonlocality}

Here we prove that any $k$-gram model with the probability distribution $p$ with $k<n/6$ cannot approximate a particular basis-enhanced 2-gram model with the probability distribution $q$ to $l_1$-distance smaller than $1/288$, i.e.
\begin{equation}\label{eq:l1}
\sum_{\bm{x}}\left\lvert p\left(\bm{x}\right)-q\left(\bm{x}\right)\right\rvert \ge 1/288.
\end{equation}
For simplicity, we assume $n=3+4k$. The key to proving the separation between the quantum extension and its classical counterpart is through a Bell test of the GHZ state through measurements in the $X$ and $Y$ bases. By measuring the remaining qubits, we obtain a GHZ state up to three single qubit Clifford gates. However, as we restricted measurement to the $X$ and $Y$ bases, this does not always hold. The following lemma gives the probability of still having nonlocality.
\begin{lemma}
The probability of measuring the remaining qubits to get a GHZ state up to Pauli and $S$ gates is larger than $1/9$. We call this measurement a \emph{GHZ-type measurement}.
\end{lemma}
\begin{proof}
Suppose the measurement basis and results for the remaining qubits are $b_1s_1b_2s_2\cdots $ with equal probability. Then the resulting state is
\begin{equation}
\sigma_1C_1\otimes \sigma_3 \otimes \sigma_2C_2 \ket{\text{GHZ}},
\end{equation}
where $\sigma_i$ is a Pauli matrix and
\begin{eqnarray}
\notag C_1&=&HS^{s_1}H\cdot S^{s_2}\cdot  \cdots \cdot HS^{s_{2k-1}}H \cdot S^{s_{2k}},\\
C_2&=&HS^{s_{2k+1}}H\cdot S^{s_{2k+2}}\cdot  \cdots \cdot HS^{s_{4k-1}}H \cdot S^{s_{4k}}.
\end{eqnarray}
We only need to prove that the probability of $C_1$ equaling $I$ or $S$ up to a Pauli matrix is at least $1/3$.
 
 All of the single qubit Clifford gates can be represented as permutations in $S_3$ among single qubit Pauli matrices up to an unimportant phase factor. $HSH$ and $S$ can be regarded as $\left(12\right)$ and $\left(23\right)$, which are generators of $S_3$. Starting from $I$, each time we apply $HS^{s_{2i-1}}H\cdot S^{s_{2i}}$ with probability $1/4$ for all of the choices of $s_{2i-1},s_{2i}$, we obtain a random walk among the 6 group elements of $S_3$. The transfer matrix is:
 \begin{equation}
\renewcommand\arraystretch{1.2}{\begin{pmatrix}
 \frac{1}{4} & \frac{1}{4} & \frac{1}{4} & 0 & \frac{1}{4} & 0 \\
 \frac{1}{4} & \frac{1}{4} & \frac{1}{4} & 0 & \frac{1}{4} & 0 \\
 \frac{1}{4} & 0 & \frac{1}{4} & \frac{1}{4} & 0 & \frac{1}{4} \\
 \frac{1}{4} & 0 & \frac{1}{4} & \frac{1}{4} & 0 & \frac{1}{4} \\
 0 & \frac{1}{4} & 0 & \frac{1}{4} & \frac{1}{4} & \frac{1}{4} \\
 0 & \frac{1}{4} & 0 & \frac{1}{4} & \frac{1}{4} & \frac{1}{4} \\
\end{pmatrix}}.
\end{equation}
By solving for the eigenstates and eigenvalues, and choosing an initial state of $\left(1,0,0,0,0,0\right)$, we find that after $k$ steps the probability to get $I$ and $S$ up to Pauli operators (which are $\left(1\right)$ and $\left(12\right)$ in $S_3$) is given by
\begin{equation}
\frac{1}{3}+\frac{2}{3}4^{-k}>\frac{1}{3},
\end{equation}
which proves the lemma.
\end{proof}

\begin{lemma}
For distributions $p$ and $q$, and any positive number $c$,
\begin{equation}
\sum_{\bm{x}}\left|p\left(\bm{x}\right)-cq\left(\bm{x}\right)\right|\ge\frac{\sum_{\bm{x}}|p\left(\bm{x}\right)-q\left(\bm{x}\right)|}{2}.
\end{equation}\end{lemma}
\begin{proof}
Denote $\delta=\sum_{\bm{x}}\left|p\left(\bm{x}\right)-q\left(\bm{x}\right)\right|$. We consider the following 2 cases:
\begin{itemize}

\item $\left|1-c\right|\ge\delta/2$:
\begin{eqnarray}
\notag \sum_{\bm{x}}\left|p\left(\bm{x}\right)-cq\left(\bm{x}\right)\right| & \ge & \left|\sum_{\bm{x}}p\left(\bm{x}\right)-cq\left(\bm{x}\right)\right|\\
\notag&=& \left|1-c\right|\\
&\ge&\delta/2.
\end{eqnarray}

\item $\left|1-c\right|\le\delta/2$:
\begin{eqnarray}
\notag \sum_{\bm{x}}\left|p\left(\bm{x}\right)-cq\left(\bm{x}\right)\right| & \ge & \sum_{\bm{x}}\left|p\left(\bm{x}\right)-q\left(\bm{x}\right)\right|-\left|\left(1-c\right)q\left(\bm{x}\right)\right|\\
\notag&=& \delta-\left|1-c\right|\\
&\ge&\delta/2.
\end{eqnarray}\end{itemize}\end{proof}

\begin{lemma}
Denote the measurement bases and results for the 3 chosen qubits as $\bm{a}=\left(a_1,a_2,a_3\right)$ and $\bm{t}=\left(t_1,t_2,t_3\right)$, respectively. Then,
\begin{equation}
\sum_{\bm{a},\bm{t}\in\operatorname{sol}_{\bm{a}}}q\left(\bm{a}\right)p\left(\bm{t}\mid\bm{a},\text{GHZ-type}\right)\le \frac{7}{8},
\end{equation}
where $\operatorname{sol}_{\mathbf a}$ means $\bm{t}$ satisfy Eq.~\eqref{eq:nonlocality} up to flips of some $b_i$ and $s_i$ determined by the GHZ-type measurement basis and results.
\end{lemma}
\begin{proof}
The probability distribution in Eq.~\eqref{eq:HVT} could also be understood as the following: there is a probability distribution $p\left(\lambda\right)$ and each $\lambda$ determines $t_i$, i.e. $t_i=t_i\left(s_i,\lambda\right)$. Because the GHZ test cannot be described by a local hidden variable theory, there exists at least 1 assignment of $\bm{a}$ given $\lambda$ such that $\bm{t}=\bm{t}\left(\bm{s},\lambda\right)\notin \operatorname{sol}_{\bm{a}}$. Assuming this and $q\left(\bm{a}\right)=1/8$ we have that:
\begin{eqnarray}
\notag\sum_{\bm{a},\bm{t}\in \operatorname{sol}_{\bm{a}}}q\left(\bm{a}\right)p\left(\bm{t}\mid\bm{a}\right)&=&
\frac{1}{8}\sum_{\bm{a}}\sum_{\bm{t}\in \operatorname{sol}_{\bm{a}}}p\left(\bm{t}\mid\bm{a}\right)\\
\notag &=&
\frac{1}{8}\int_\lambda p\left(\lambda\right)\sum_{\bm{a}}\bm{1}\left(\bm{t}\left(\lambda,\bm{a}\right)\in \operatorname{sol}\right)\\
&\le&\frac{7}{8}
\end{eqnarray}
where $\bm{1}\left(\cdot\right)$ is the indicator function which is 1 if the condition $\cdot$ holds, and is otherwise 0.
\end{proof}

\begin{figure*}[ht]
  \begin{center}
    \includegraphics[width=\textwidth]{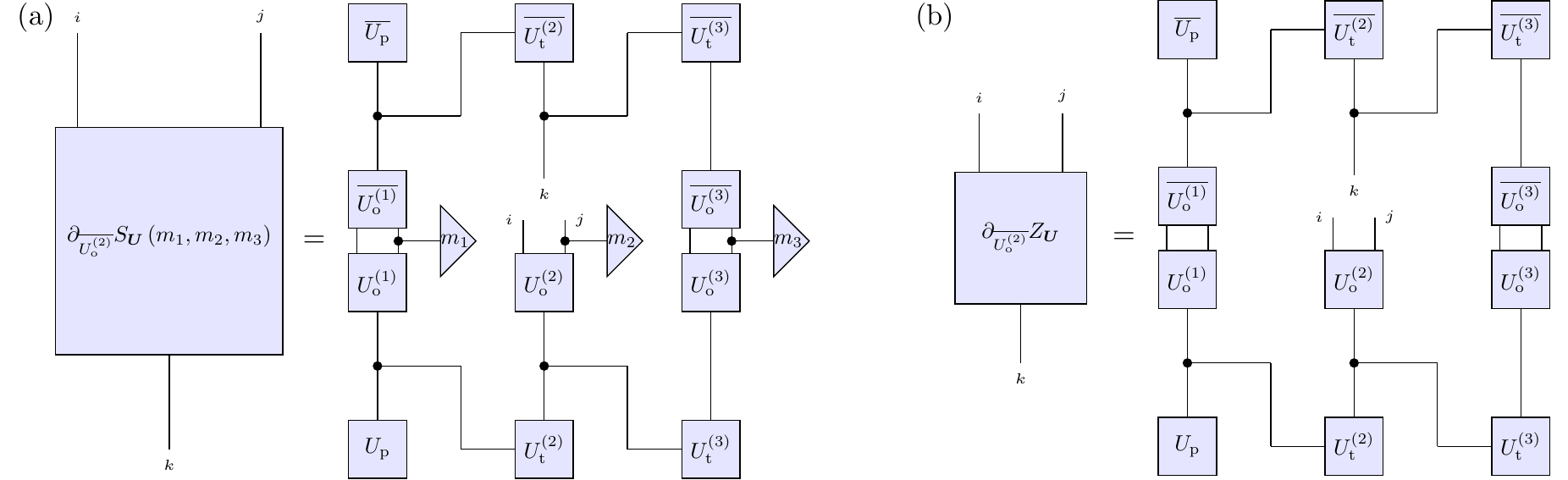}
    \caption{(a) The derivative of $S_{\bm{U}}\left(\bm{m}\right)$  with respect to $\overline{U_{\textrm{o}}^{\left(2\right)}}$, using the network given in Fig.~\ref{fig:q2gram_circuit}(b) as an illustrative example. (b) The derivative of $Z_{\bm{U}}$ with respect to $\overline{U_{\textrm{o}}^{\left(2\right)}}$, using the network given in Fig.~\ref{fig:q2gram_circuit}(b) as an illustrative example.\label{fig:dv}}
  \end{center}
\end{figure*}

\begin{figure*}[ht]
  \begin{center}
    \includegraphics[scale=0.275]{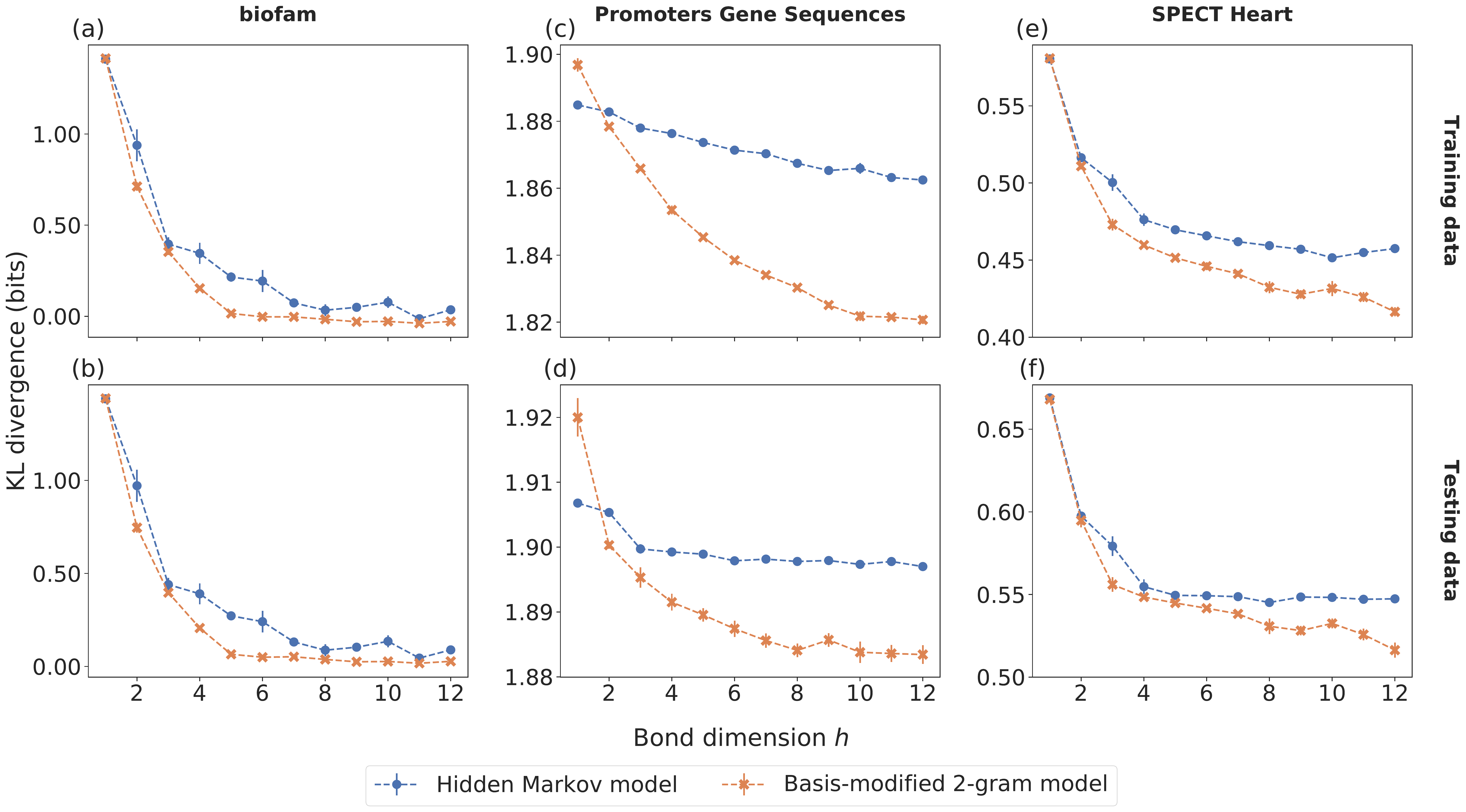}
    \caption{Shown are the average performances over ten trials of the classical HMM (blue circles) and the basis-enhanced $2$-gram model (orange crosses) on the biofam ((a) and (b)), Promoter Gene Sequences ((c) and (d)), and SPECT Heart ((e) and (f)) data sets. The first row plots the performance on the training data, and the second the performance on withheld testing data. The basis-enhanced $2$-gram, on average, performed better than the classical model even on the simple biofam data set, implying more consistent performance in the basis-enhanced $2$-gram model. Error bars denote one standard error of the mean over ten trials. Dashed lines are to aid the eye.\label{fig:avg_performance_plot}}
  \end{center}
\end{figure*}

Combining the above lemmas, we now show that Eq.~\eqref{eq:l1} holds. First:
%\begin{proof}

\begin{widetext}
\begin{eqnarray}
\notag\sum_{\bm{a},\bm{b},\bm{s},\bm{t}}\left|p\left(\bm{a},\bm{b},\bm{s},\bm{t}\right)-q\left(\bm{a},\bm{b},\bm{s},\bm{t}\right)\right|&=&
\sum_{\bm{a},\bm{b},\bm{s},\bm{t}}\left|p\left(\bm{b},\bm{s}\right)p\left(\bm{a},\bm{t}\mid\bm{b},\bm{s}\right)-q\left(\bm{b},\bm{s}\right)q\left(\bm{a},\bm{t}\mid\bm{b},\bm{s}\right)\right|\\
\notag
&=&
\left(\sum_{\bm{b},\bm{s}\in\text{GHZ-type}}+\sum_{\bm{b},\bm{s}\notin\text{GHZ-type}}\right)\sum_{\bm{a},\bm{t}}\left|p\left(\bm{b},\bm{s}\right)p\left(\bm{a},\bm{t}\mid\bm{b},\bm{s}\right)-q\left(\bm{b},\bm{s}\right)q\left(\bm{a},\bm{t}\mid\bm{b},\bm{s}\right)\right|\\
\notag &\ge& 
\sum_{\bm{b},\bm{s}\in\text{GHZ-type}}
\sum_{\bm{a},\bm{t}}\left|p\left(\bm{b},\bm{s}\right)p\left(\bm{a},\bm{t}\mid\bm{b},\bm{s}\right)-q\left(\bm{b},\bm{s}\right)q\left(\bm{a},\bm{t}\mid\bm{b},\bm{s}\right)\right|\\
\notag
&\ge&\sum_{\bm{b},\bm{s}\in\text{GHZ-type}}p\left(\bm{b},\bm{s}\right)
\sum_{\bm{a},\bm{t}}\left|p\left(\bm{a},\bm{t}\mid\bm{b},\bm{s}\right)-\frac{q\left(\bm{b},\bm{s}\right)}{p\left(\bm{b},\bm{s}\right)}q\left(\bm{a},\bm{t}\mid\bm{b},\bm{s}\right)\right|\\
\notag&\ge&
\sum_{\bm{b},\bm{s}\in\text{GHZ-type}}\frac{p\left(\bm{b},\bm{s}\right)}{2}
\sum_{\bm{a},\bm{t}}\left|p\left(\bm{a},\bm{t}\mid\bm{b},\bm{s}\right)-q\left(\bm{a},\bm{t}\mid\bm{b},\bm{s}\right)\right|\\
\notag&\ge&\sum_{\bm{b},\bm{s}\in\text{GHZ-type}}\frac{p\left(\bm{b},\bm{s}\right)}{2}
\min_{\bm{b},\bm{s}\in\text{GHZ-type}}\sum_{\bm{a},\bm{t}}\left|p\left(\mathbf{a,t}\mid\bm{b},\bm{s}\right)-q\left(\bm{a},\bm{t}\mid\bm{b},\bm{s}\right)\right|\\
&=&\frac{1}{18}\min_{\bm{b},\bm{s}\in\text{GHZ-type}}\sum_{\bm{a},\bm{t}}\left|p\left(\bm{a},\bm{t}\mid\bm{b},\bm{s}\right)-q\left(\bm{a},\bm{t}\mid\bm{b},\bm{s}\right)\right|.
\end{eqnarray}
Making the minimization over GHZ-types implicit and noting that $q\left(\bm{t}\mid\bm{a},\text{GHZ-type}\right)=1$ for $\bm{t}\in\operatorname{sol}_{\bm{a}}$ finally yields:
\begin{eqnarray}
\notag \sum_{\bm{a},\bm{t}}\left|p\left(\bm{a},\bm{t}\right)-q\left(\bm{a},\bm{t}\right)\right|&=&\sum_{\bm{a},\bm{t}}\left|p\left(\bm{a}\right)p\left(\bm{t}\mid\bm{a}\right)-q\left(\bm{a}\right)q\left(\bm{t}\mid\bm{a}\right)\right|\\
\notag &=&\sum_{\bm{a}}q\left(\bm{a}\right)
\sum_{\bm{t}}\left|q\left(\bm{t}\mid\bm{a}\right)-\frac{p\left(\bm{a}\right)}{q\left(\bm{a}\right)}p\left(\bm{t}\mid\bm{a}\right)\right|\\
\notag &\ge&\frac{1}{2}\sum_{\bm{a}}q\left(\bm{a}\right)
\sum_{\bm{t}}\left|q\left(\bm{t}\mid\bm{a}\right)-p\left(\bm{t}\mid\bm{a}\right)\right|\\
\notag &=&\frac{1}{2} \left(1- \sum_{\bm{a,t}\in\operatorname{sol}_{\bm{a}}}q\left(\bm{a}\right)p\left(\bm{t}\mid\bm{a}\right)\right)\\
&\ge&\frac{1}{16}.
\end{eqnarray}
With the last inequality, we arrive at the separation of $1/288$ under the $l_1$-distance.
\end{widetext}

\section{HHM to Simulate Fig.~\ref{fig:quantum2gram}(a)}\label{app:hmm}
The process described in Fig.~\ref{fig:quantum2gram} could also be understood by an HMM model described in Sec.~\ref{sec:hmm}. The total number of possible quantum states involved is an upper bound for the number of hidden variables in HMM. The state stored in the second qubit of each pair in Fig.~\ref{fig:quantum2gram}, with the evolution driven by measuring the first qubit in each pair, only involves single qubit stabilizer states. The total number is thus $6$ ($\ket0,\ket1,\ket+,\ket-,\ket{+_i},\ket{-_i}$).

\section{Supplemental Numerics Figures}\label{app:avg_perf}

In Fig.~\ref{fig:dv}, we plot the tensor network representation of the derivative of the loss function Eq.~\eqref{eq:derivative_of_loss}. Furthermore, in Fig.~\ref{fig:avg_performance_plot}, we plot the average performance of both the trained HMM and basis-enhanced $2$-gram model. Note that the performance separation between quantum and classical models persists even when considering the average performance over many runs. A slight separation is observed even for the biofam data set, implying more consistent performance in the basis-enhanced $2$-gram model.

%\bibliographystyle{naturemag}
%\bibliography{ref}

\end{document}